\documentclass{article}
\usepackage{amsmath,amssymb,amsfonts,amsthm}
\usepackage[all]{xy}
\usepackage[breaklinks=true]{hyperref}
\usepackage{graphics}
\usepackage{color}

\usepackage{rotating}
\usepackage{hyperref}

\begin{document}

	\theoremstyle{plain}
\newtheorem{theorem}{Theorem}[section]
\newtheorem{proposition}[theorem]{Proposition}
\newtheorem{corollary}[theorem]{Corollary}

\newtheorem{observation}[theorem]{Observation}

\newtheorem{theoremapp}{Theorem A.\!\!}
\newtheorem{propositionapp}{Proposition A.\!\!}

\newtheorem{corollaryapp}{Corollary A.\!\!}
\newtheorem{observationapp}{Observation A.\!\!}
\newtheorem{lemmaapp}{Lemma A.\!\!}

\theoremstyle{definition}
\newtheorem{definition}[theorem]{Definition}
\newtheorem{dfn}[theorem]{definition}
\newtheorem{remark}[theorem]{Remark}
\newtheorem{example}[theorem]{Example}
\newtheorem{examples}[theorem]{Examples}

\newtheorem{definitionapp}{Definition A.\!\!}
\newtheorem{dfnapp}{definitionapp}
\newtheorem{remarkapp}{Remark A.\!\!}
\newtheorem{exampleapp}{Example A.\!\!}

\newdir^{ (}{{}*!/-5pt/@^{(}}

\renewcommand{\theequation}{\thesubsection.\arabic{equation}}

\def\proof {{\it Proof.}\hspace{7pt}}

\def\endofproof{\hfill{$\square$}\\}
\def\dfn{Definition}
\def\rmk{Remark}
\def\Cech{{{\v C}ech}}

\def\Loo{{$L_\infty$}}

\pagestyle{myheadings}

\title{A higher Chern-Weil derivation of AKSZ $\sigma$-models}
\author{Domenico Fiorenza, Christopher L.\ Rogers, and Urs Schreiber}

\maketitle	

\begin{abstract}
Chern-Weil theory provides for each invariant
polynomial on a Lie algebra $\mathfrak{g}$ a map from
$\mathfrak{g}$-connections to differential cocycles whose
volume holonomy is the corresponding 
Chern-Simons theory action functional.
Kotov and Strobl have observed that this naturally generalizes
from Lie algebras to dg-manifolds and dg-bundles and that the Chern-Simons
action functional associated this way to an $n$-symplectic manifold
is the action functional of the AKSZ $\sigma$-model whose
target space is the given $n$-symplectic manifold (examples of this are the Poisson
$\sigma$-model or the Courant $\sigma$-model, including ordinary Chern-Simons
theory, or higher dimensional abelian Chern-Simons theory). Here we 
show how, within the framework of the higher Chern-Weil theory in smooth
$\infty$-groupoids, this result can be naturally recovered and enhanced to a
morphism of higher stacks, the same way as ordinary Chern-Simons theory is
enhanced to a morphism from the stack of principal $G$-bundles with
connections to the 3-stack of line 3-bundles with connections.
\end{abstract}

\newpage

\tableofcontents

\newpage

\section{Introduction.}

The class of topological field theories
known as \emph{AKSZ $\sigma$-models} \cite{AKSZ} 
contains in dimension 3 ordinary Chern-Simons theory 
(see \cite{FreedCS} for a comprehensive review)
as well as its Lie algebroid generalization 
(the \emph{Courant $\sigma$-model} \cite{Ikeda02, HoPa04,RoytenbergAKSZ}), and
in dimension 2 the Poisson $\sigma$-model \cite{Ikeda93,ScSt94}
(see \cite{ScSt94b,CattaneoFelder} for a review). It is 
therefore clear that the AKSZ construction 
is \emph{some} sort of generalized Chern-Simons theory.
That this is indeed true, has been first formalized and rigorously established in \cite{KoSt07},
as a particular case of a general Chern-Weil-type construction of
characteristic classes for Q-bundles (see also \cite{KoSt10} for a review). Here we show how, within the
framework of the $\infty$-Chern-Weil homomorphism of \cite{FSSI, survey},
this result can be naturally recovered and enhanced to a morphism of higher
stacks, the same way as ordinary Chern-Simons theory is enhanced to a
morphism 
from the stack of principal $G$-bundles with connections to the 3-stack of
line 3-bundles with connections.

%

Our discussion proceeds from the observation that 
the standard Chern-Simons action
functional has a systematic origin in 
Chern-Weil theory (see for instance \cite{GHV} for 
a classical textbook treatment and \cite{HopkinsSinger} for 
the refinement to differential cohomology that we need here):
 
The refined Chern-Weil homomorphism assigns to any invariant 
polynomial $\langle -\rangle : \mathfrak{g}^{\otimes_k} \to \mathbb{R}$ 
on a Lie algebra $\mathfrak{g}$ 
of compact type a map that sends $\mathfrak{g}$-connections $\nabla$ 
on a smooth manifold $X$ to  cocycles 
$[\hat {\mathbf{p}}_{\langle-\rangle}(\nabla)] \in
H^{2k}_{\mathrm{diff}}(X)$
in \emph{ordinary differential cohomology}.  
These differential cocycles 
refine the \emph{curvature characteristic class} 
$[\langle F_\nabla \rangle] \in H_{dR}^{2k}(X)$ in de Rham cohomology
to a fully fledged \emph{line $(2k-1)$-bundle with connection}, also known as
a \emph{bundle $(2k-2)$-gerbe with connection}. 
And just as an ordinary line bundle (a ``line 1-bundle'') with connection
assigns holonomy to curves, so a line $n$-bundle with connection
assigns holonomy 
$\mathrm{hol}_{\hat {\mathbf{p}}}(\Sigma)$ 
to $n$-dimensional trajectories $\Sigma \to X$. 
For the special case where $\langle -\rangle$ 
is the Killing form polynomial 
and $X = \Sigma$ with $\dim\Sigma = 3$ one finds that this 
volume holonomy map
$\nabla \mapsto \mathrm{hol}_{\hat {\mathrm{p}}_{\langle -\rangle}(\nabla)}(\Sigma)$
is precisely the standard Chern-Simons action functional.
Similarly, for $\langle -\rangle$ any higher invariant polynomial
this holonomy action functional has as Lagrangian the 
corresponding higher Chern-Simons form. 
In summary, this means that Chern-Simons-type action functionals
on Lie algebra-valued connections are the images of the 
refined Chern-Weil homomorphism.
\par
In previous work \cite{survey,FSSI} a generalization of the 
Chern-Weil homomorphism to \emph{higher} (``derived'')
differential geometry has been considered. In this framework,
smooth manifolds are generalized 
first to orbifolds, then to general Lie groupoids, 
to Lie 2-groupoids and finally to smooth $\infty$-groupoids 
(smooth $\infty$-stacks), while Lie algebras are generalized to
Lie $2$-algebras etc., up to $L_\infty$-algebras 
and more generally to Lie $m$-algebroids and 
finally to $L_\infty$-algebroids.\footnote{See \cite{KoSt07} for a treatment of
the higher Chern-Weil homomorphism from the equivalent point of view of
Q-manifolds.} For $\mathfrak{a}$ 
any $L_\infty$-algebroid, one has a natural notion of $\mathfrak{a}$-valued 
$\infty$-connections on $\exp(\mathfrak{a})$-principal 
smooth $\infty$-bundles
(where $\exp(\mathfrak{a})$ is a smooth $\infty$-groupoid obtained by
Lie integration from $\mathfrak{a}$).
By analyzing the abstractly defined higher Chern-Weil homomorphism
in this context one finds a direct higher analog of the above
situation: there is a notion of invariant polynomials $\langle-\rangle$ on 
an $L_\infty$-algebroid $\mathfrak{a}$ 
and these induce maps from $\mathfrak{a}$-valued $\infty$-connections
to line $n$-bundles with connections as before \cite{SSSI, FSSI}. 
The corresponding class of action functionals we call 
\emph{$\infty$-Chern-Simons theory}. 
\par
This construction drastically simplifies when one restricts attention to
\emph{trivial} $\infty$-bundles with (nontrivial) 
$\mathfrak{a}$-connections. Over a smooth manifold $\Sigma$
these are simply given by dg-algebra homomorphisms 
\[
  A : \mathrm{W}(\mathfrak{a}) \to \Omega^\bullet(\Sigma)
  \,,
\]
where $\mathrm{W}(\mathfrak{a})$ is the \emph{Weil algebra} of the 
$L_\infty$-algebroid $\mathfrak{a}$ \cite{KoSt07,SSSI}, 
and $\Omega^\bullet(\Sigma)$ is the de Rham algebra of $\Sigma$ 
(which is indeed the Weil algebra of $\Sigma$ thought of as an 
$L_\infty$-algebroid concentrated in degree 0). 
Then for $\langle-\rangle \in \mathrm{W}(\mathfrak{a})$ an invariant
polynomial, the corresponding $\infty$-Chern-Weil homomorphism
is presented by a choice of ``Chern-Simons element''
$\mathrm{cs} \in \mathrm{W}(\mathfrak{a})$, which exhibits the
\emph{transgression} of $\langle-\rangle$ to an $L_\infty$-cocycle
(the higher analog of a cocycle in Lie algebra cohomology):
the dg-morphism $A$ naturally maps the Chern-Simons element 
$\mathrm{cs}$ of $A$ to a differential form $\mathrm{cs}(A) \in \Omega^\bullet(\Sigma)$  and its integral is the corresponding 
$\infty$-Chern-Simons action functional $S_{\langle -\rangle}$
$$
  S_{\langle-\rangle}
  :
  A \mapsto \mathrm{hol}_{\hat {\mathbf{p}}_{\langle-\rangle}(A)}(\Sigma)
    = 
    \int_\Sigma \mathrm{cs}_{\langle-\rangle}(A)
  \,.
$$
\par
Even though trivial $\infty$-bundles with $\mathfrak{a}$-connections
are a very particular subcase of the general $\infty$-Chern-Weil
theory, they are rich enough to contain AKSZ theory. Namely, here
we show that a symplectic dg-manifold of grade $n$
-- which is  the geometrical datum of 
the target space defining an AKSZ $\sigma$-model --
is naturally equivalent to an $L_\infty$-algebroid 
$\mathfrak{P}$ endowed with a
quadratic and non-degenerate invariant polynomial $\omega$
of grade $n$. 
Moreover, under this identification the canonical Hamiltonian
$\pi$ on the symplectic target dg-manifold is identified 
as an $L_\infty$-cocycle on $\mathfrak{P}$. Finally, the
invariant polynomial $\omega$ is naturally in transgression with the
cocycle $\pi$ via a Chern-Simons element $\mathrm{cs}_\omega$ that turns 
out to be the Lagrangian of the AKSZ $\sigma$-model \cite{KoSt07}:
$$
  \int_\Sigma L_{\mathrm{AKSZ}}(A) 
   = 
  \int_\Sigma \mathrm{cs}_{\omega}(A)
  \,.
$$
(An explicit description of $L_{\mathrm{AKSZ}}$ is given below
in def. \ref{TheAKSZAction})

In summary this means that we find the following
dictionary of concepts:\\
\begin{center}
\begin{tabular}{c|c|c}
  {\bf Chern-Weil theory} && {\bf AKSZ theory}
  \\
  \hline
  &&
  \\
  cocycle & $\pi$ &  Hamiltonian
  \\
   &&
  \\
  transgression element & $\mathrm{cs}$ & Lagrangian
  \\
    &&
  \\
  invariant polynomial & $\omega$ & symplectic structure
  \\
  &&
  \\
  \hline
\end{tabular}
\end{center}
More precisely, we (explain and then) prove here the following theorem:
\begin{theorem}
  \label{AKSZIsCW}
  For $(\mathfrak{P}, \omega)$ an $L_\infty$-algebroid with
  a quadratic non-degenerate invariant polynomial, the
  corresponding $\infty$-Chern-Weil homomorphism
  $$
    \nabla \mapsto \mathrm{hol}_{\hat {\mathbf{p}}_\omega(\nabla)}(\Sigma)
  $$
  sends $\mathfrak{P}$-valued $\infty$-connections $\nabla$
  to their corresponding exponentiated AKSZ action:
  $$
 \mathrm{hol}_{\hat {\mathbf{p}}_\omega(\nabla)}(\Sigma)= \int_\Sigma L_{\mathrm{AKSZ}}(\nabla)
    \,.
  $$
\end{theorem}
The local differential form data involved in this statement 
is at the focus of attention in this article here and 
contained in proposition \ref{TheAKSZActionFromCS} below. 
We indicate the global  aspects of the construction 
in Section \ref{section.CS_theory}.
The more abstract higher Chern-Weil theoretic interpretation of 
AKSZ $\sigma$-models implies various further constructions and
generalizations. We close in Section \ref{section.generalizations} by giving 
an outlook on these.
\smallskip

After a preliminary version of this article had appeared on arXiv we were informed by Alexei Kotov and Thomas Strobl that the main constructions from Sections 2-4 are presented in \cite{KoSt07} in the equivalent framework of Q-manifolds. For the reader's convenience, here is a basic dictionary between the notations used in this article and those in  \cite{KoSt07}: what Kotov and Strobl call a Q-manifold is by definition a dg-manifold in our terminology; our homological vector field $v$ is their $Q$ and our $d_{\mathrm{W}(\mathfrak{a})}$ is their
total differential (see proposition \ref{WeilAlgebraExplicit} below); our $A$ in definition \ref{definition.a-valued-differential-form} is their chain map $f^*$; our proposition \ref{CanonicalCocycleOfSymplecticLieAlgebroid} is their lemma 4.5; our proposition \ref{TheCSElement} is their lemma 4.6, with the Chern-Simons element we denote $\mathrm{cs}$ called
$\hat{\alpha}$ there; finally, our proposition \ref{TheAKSZActionFromCS} and corollary \ref{KS4.4} are their theorem 4.4.
\medskip
\medskip

\noindent{\bf Acknowledgement.} We thank Dmitry Roytenberg
for helpful conversation at an early stage of this project, Alexei Kotov and Thomas Strobl for having pointed our attention to \cite{KoSt07}, and the Referee for comments and suggestions on the first
version of this article.
CLR acknowledges support by a Junior Research Fellowship from the Erwin
Schr\"{o}dinger International Institute for Mathematical Physics.

\medskip
\medskip

\section{A reminder of AKSZ theory}
\label{section.reminder_AKSZ}

The first half of the seminal article \cite{AKSZ} presented some key
observations on, what from a modern perspective would be
called, symplectic \emph{derived geometry} \cite{LurieSpaces}
in its variant of symplectic \emph{dg-geometry} 
\cite{ToenVezzosi}. In its second half, it 
describes the role of such symplectic 
dg-geometry in quantum field theory
in general, and $\sigma$-model theory in particular.

In this section we briefly review some basics
in order to establish the context for our discussion.

\subsection{Symplectic dg-geometry}
\label{section.symplectic_dg-geometry}

In \emph{higher differential geometry}, smooth manifolds are
generalized first to orbifolds -- which are special Lie groupoids --
then to higher Lie groupoids: smooth $\infty$-groupoids \cite{survey}.
Moreover, in \emph{derived differential geometry}, the function
algebras are generalized to
\emph{smooth $\infty$-algebras} \cite{Spivak, Stel}.
All of these ingredients have \emph{presentations}  
in terms of compound structures in ordinary differential geometry.
There is a bit of theory involved in exactly how  these presentations
model the general abstract theory, but the main statement that 
we want to discuss here can be described already in a rather simple-minded
setup. 

Therefore, here we shall be content with the following simple
definitions of what might be called 
\emph{affine smooth graded manifolds} and 
\emph{affine smooth dg-manifolds}. 
Despite their simplicity
these definitions capture in a precise sense all the relevant structure: namely the
\emph{local} smooth structure. Globalizations of these definitions
can be obtained, if desired, by general abstract constructions.
We give some outlook on this in section 
\ref{section.generalizations}.

\begin{definition}
  \label{GradedManifold}
  The category of \emph{affine smooth $\mathbb{N}$-graded manifolds} 
  -- here called \emph{smooth graded manifolds} for short --
  is the full subcategory
  $$
    \mathrm{SmoothGrMfd} \subset \mathrm{GrAlg}_{\mathbb{R}}^{\mathrm{op}}
  $$
  of the opposite category of 
  $\mathbb{N}$-graded-commutative
  $\mathbb{R}$-algebras on those isomorphic to Grassmann algebras
  of the form
  $$
    \wedge^\bullet_{C^\infty(X_0)} \Gamma(V^*)
    \,,
  $$
  where $X_0$ is an ordinary smooth manifold, $V \to X_0$ is an 
  $\mathbb{N}$-graded smooth vector bundle over $X_0$ degreewise of 
  finite rank, and
  $\Gamma(V^*)$ is the graded $C^\infty(X)$-module of
  smooth sections of the dual bundle.
  
  For a smooth graded manifold $X \in \mathrm{SmoothGrMfd}$, we write 
  $C^\infty(X) \in \mathrm{cdgAlg}_{\mathbb{R}}$ for its corresponding
  dg-algebra \emph{of functions}.
\end{definition}
\noindent {\bf Remarks.}
\begin{itemize}
  \item The full subcategory of these objects is equivalent to that
  of all objects isomorphic to one of this form. We may therefore
  use both points of view interchangeably. 
  \item Much of the theory works just as well 
   when $V$ is allowed to be $\mathbb{Z}$-graded. This is the
   case that genuinely corresponds to \emph{derived} 
   (instead of just higher)
   differential geometry. An important class of examples for this case are
   BV-BRST complexes which motivate much of the literature.
   For the purpose of this short note, we shall be content with the 
   $\mathbb{N}$-graded case. 
  \item
   For an $\mathbb{N}$-graded $C^\infty(X_0)$-module $\Gamma(V^*)$ we have
   $$
     \wedge^\bullet_{C^\infty} \Gamma(V^*)
     =
     C^\infty(X_0)
      \oplus
     \Gamma(V_0^*)
       \oplus
     \left(
       \Gamma(V_0^*) \wedge_{C^\infty(X_0)} \Gamma(V_0^*)
       \oplus
       \Gamma(V_1^*)
     \right)
     \oplus
     \cdots
     \,,
   $$
   with the leftmost summand in degree 0, the next one in degree 1,
   and so on.
   \item There is a canonical functor
   $$
     \mathrm{SmoothMfd} \hookrightarrow \mathrm{SmthGrMfd}
   $$
   which identifies an ordinary smooth manifold $X$ with the 
   smooth graded manifold whose function algebra is the 
   ordinary algebra of smooth functions $C^\infty(X_0) := C^\infty(X)$
   regarded as a graded algebra concentrated in degree 0.
   This functor is full and faithful and hence exhibits a full
   subcategory.
\end{itemize}
All the standard notions of differental geometry apply to 
differential graded geometry. 
For instance for $X \in \mathrm{SmoothGrMfd}$, 
there is the graded vector space $\Gamma(T X)$ of vector fields on $X$, 
where a vector field is identified with a
graded \emph{derivation} $v : C^\infty(X) \to C^\infty(X)$.
This is naturally a graded (super) Lie algebra with 
super Lie bracket the graded commutator of derivations.
Notice that for $v \in \Gamma(T X)$ of odd degree we have
$[v,v] = v\circ v + v \circ v = 2 v^2 : C^\infty(X) \to C^\infty(X)$.
\begin{definition}
  \label{DGManifolds}
  The category  of (affine, $\mathbb{N}$-graded) 
  \emph{smooth differential-graded manifolds}
  is the full subcategory 
  $$
     \mathrm{SmoothDgMfd} \subset \mathrm{cdgAlg}_{\mathbb{R}}^{\mathrm{op}}
  $$ 
  of the opposite of differential graded-commutative $\mathbb{R}$-algebras
  on those objects whose underlying graded algebra comes from
  $\mathrm{SmoothGrMfd}$.
\end{definition}
This is equivalently the category whose objects 
are  pairs $(X,v)$ consisting of a smooth graded manifold
$X \in \mathrm{SmoothGrMfd}$
and a grade 1 vector field $v \in \Gamma(T X)$, 
such that $[v,v] = 0$, and whose morphisms $(X_1,v_1) \to (X_2, v_2)$ are 
morphisms $f : X_1 \to X_2$ such that $v_1 \circ f^*  = f^* \circ v_2$.
\begin{remark}
  \label{SmoothDgAlgebras}
  The dg-algebras appearing here are special in that their
  degree-0 algebra is naturally not just
  an $\mathbb{R}$-algebra, but a \emph{smooth algebra}
  (a ``$C^\infty$-ring'', see \cite{Stel} for review and
  discussion). In a more theoretical account than we want to 
  present here, we would use the corresponding more general notion of 
  \emph{smooth dg-algebras}. For our present purposes, 
  this will only briefly play a role in def. \ref{WeilAlgebraAbstract}
  below.
\end{remark}
\begin{definition}
  \label{deRhamComplex}
  The \emph{de Rham complex functor}
  $$
    \Omega^\bullet(-) 
      : 
    \mathrm{SmoothGrMfd}
     \to 
    \mathrm{cdgAlg}^{\mathrm{op}}_{\mathbb{R}}
  $$
  sends a dg-manifold $X$ with 
  $C^\infty(X) \simeq \wedge^\bullet_{C^\infty(X_0)}
  \Gamma(V^*)$ to the 
  Grassmann algebra over $C^\infty(X_0)$ on the graded
  $C^\infty(X_0)$-module
  $$
    \Gamma(T^* X) \oplus \Gamma(V^*) \oplus \Gamma(V^*[-1])
    \,,
  $$
  where $\Gamma(T^* X)$ denotes the ordinary smooth 1-form fields
  on $X_0$ and where $V^*[-1]$ is $V^*$ with the grades \emph{increased}
  by one. This is equipped with the differential $\mathbf{d}$
  defined on generators as follows:
  \begin{itemize}
   \item $\mathbf{d}|_{C^\infty(X_0)} = d_{\mathrm{dR}}$ is the 
  ordinary de Rham differential with values in $\Gamma(T^* X)$;
    \item 
      $\mathbf{d}|_{\Gamma(V^*)} \to \Gamma(V^*[-1])$ is the degree-shift
  isomorphism
   \item and $\mathbf{d}$ vanishes on all remaining generators.
   \end{itemize}
\end{definition}
\begin{definition}
\label{TangentLieAlgebroidAsDgManifold}
Observe that $\Omega^\bullet(-) $ evidently factors through the
defining inclusion 
$\mathrm{SmoothDgMfd}$ $\hookrightarrow$ $\mathrm{cdgAlg}_{\mathbb{R}}$.
Write
$$
  \mathfrak{T}(-) : \mathrm{SmoothGrMfd} \to \mathrm{SmoothDgMfd}
$$
for this factorization.
\end{definition}
 \label{GradedTangentBundle}
The dg-space $\mathfrak{T}X$ is often called the 
\emph{shifted tangent bundle} of $X$ and denoted $T[1]X$.
\begin{observation}
  \label{FormsFromGrManifoldMaps}
  For $\Sigma$ an ordinary smooth manifold and 
  for $X$ a graded manifold corresponding to a vector bundle
  $V \to X_0$, there is a natural bijection
  $$
    \mathrm{SmoothGrMfd}(\mathfrak{T}\Sigma, X)
    \simeq
    \Omega^\bullet(\Sigma,V)
    \,
  $$
  where on the right we have the set of $V$-valued smooth differential forms
  on $\Sigma$: tuples consisting of a smooth function 
  $\phi_0 : \Sigma \to X_0$, and for each $n > 1$ 
  an ordinary differential $n$-form 
  $\phi_n \in \Omega^n(\Sigma, \phi_0^* V_{n-1})$ with values in the pullback
  bundle of $V_{n-1}$ along $\phi_0$.
\end{observation}
The standard Cartan calculus of differential geometry
generalizes directly to graded smooth manifolds.
For instance, given a vector field $v \in \Gamma(T X)$
on $X \in \mathrm{SmoothGrMfd}$, 
there is the \emph{contraction derivation}
$$
  \iota_v : \Omega^\bullet(X) \to \Omega^{\bullet}(X)
$$
on the de Rham complex of $X$, and hence the \emph{Lie derivative}
$$
  \mathcal{L}_v := [\iota_{v},\mathbf{d}] : 
  \Omega^\bullet(X) \to \Omega^\bullet(X)
  \,.
$$
\begin{definition}
  For $X \in \mathrm{SmoothGrMfd}$ the \emph{Euler
  vector field} $\epsilon \in \Gamma(T X)$ is defined over 
  any coordinate patch $U \to X$ to be
  given by the formula
  $$
    \epsilon|_U 
      := 
    \sum_a
     \mathrm{deg}(x^a) x^a \frac{\partial}{\partial x^a}
    \,,
  $$
  where $\{x^a\}$ is a basis of generators and 
  $\mathrm{deg}(x^a)$ the degree of a generator.
    The \emph{grade} of a homogeneous element $\alpha$ in 
  $\Omega^\bullet(X)$ is the unique natural number $n \in \mathbb{N}$
  with 
  $$
    \mathcal{L}_\epsilon \alpha = n \alpha
    \,.
  $$
\end{definition}
\noindent{\bf Remarks.}
\begin{itemize}
\item This implies that for $x^i$ an element of grade $n$ on $U$, the
1-form $\mathbf{d}x^i$ is also of grade $n$. This is why we speak
of \emph{grade} (as in ``graded manifold'') instead of \emph{degree} here. 
\item Since coordinate transformations on a graded manifold are
  grading-preserving, the Euler vector field is indeed
  well-defined. Note that the degree-0 coordinates 
  do not appear in the Euler vector field.
\end{itemize}
The existence of $\epsilon$ implies the following 
useful statement (amplified in \cite{RoytenbergCourant}),
which is a trivial variant of what in grade 0
would be the standard Poincar{\'e} lemma.
\begin{observation}
  \label{GradedPoincareLemma}
  On a graded manifold, every closed differential form
  $\omega$ of positive grade $n$ is exact: the form 
  $$
    \lambda := \frac{1}{n} \iota_\epsilon \omega
  $$
  satisfies
  $$
    \mathbf{d}\lambda = \omega
    \,.
  $$
\end{observation}
\begin{definition}
  \label{SymplecticDgManifold}
  A \emph{symplectic dg-manifold} of grade $n \in \mathbb{N}$ is 
  a dg-manifold $(X,v)$ equipped with 
  2-form $\omega \in \Omega^2(X)$ which is
  \begin{itemize}
    \item non-degenerate;
    \item closed;
  \end{itemize}
  as usual for symplectic forms, and
  in addition
  \begin{itemize}
    \item of grade $n$;
    \item $v$-invariant: $\mathcal{L}_v \omega = 0$.
  \end{itemize}
\end{definition}
In a local chart $U$ with coordinates $\{x^a\}$ we may find functions
$\{\omega_{ab} \in C^\infty(U)\}$ such that
$$
  \omega|_U = \frac{1}{2}  \mathbf{d}x^a\, \omega_{ab}\wedge \mathbf{d}x^b
  \,,
$$
where summation of repeated indices is implied. We say that 
$U$ is a \emph{Darboux chart} for $(X,\omega)$ if the $\omega_{ab}$
are constant.
\begin{observation}  
  The function algebra of a symplectic dg-manifold 
  $(X,\omega)$ of grade $n$
  is naturally equipped with a Poisson bracket
  $$
    \{-,-\} : 
    C^\infty(X)\otimes C^\infty(X)
    \to C^\infty(X)
  $$
  which decreases grade by $n$. On a local coordinate patch $\{x^a\}$ this
  is given by
  $$
    \{f,g\} = 
    \frac{f \reflectbox{\small $\partial$}}{x^a \reflectbox{\small $\partial$}}
    \omega^{a b}
    \frac{\partial g}{\partial x^b}
    \,,
  $$
  where $\{\omega^{a b}\}$ is the inverse matrix to $\{\omega_{a b}\}$,
  and where the graded differentiation in the left factor is to 
  be taken from the right, as indicated.
\end{observation}
\begin{definition}
  \label{Hamiltonians}
  For $\pi\in C^\infty(X)$ and $v \in \Gamma(T X)$, we say that
  \emph{$\pi$ is a Hamiltonian for $v$}, or equivalently, that
  \emph{$v$ is the Hamiltonian vector field of $\pi$}
  if
  $$
    \mathbf{d}\pi = \iota_v \omega
    \,.
  $$
  \end{definition}
Note that the convention $(-1)^{n+1} \mathbf{d}\pi = \iota_v \omega$
is also  frequently used for defining Hamiltonians in the 
context of graded geometry.
\begin{remark}
\label{PartialPiByV}
In a local coordinate chart $\{x^a\}$ the defining equation
$\mathbf{d}\pi = \iota_v \omega$
becomes 
$$
    \mathbf{d}x^a \frac{\partial \pi}{\partial x^a}
    =
    \omega_{a b} v^a \wedge \mathbf{d} x^b
    = 
    \omega_{a b} \mathbf{d}x^a \wedge v^b
  \,,
$$
implying that 
$$
  \omega_{ab}v^b = \frac{\partial \pi}{\partial x^a}
  \,.
$$  
\end{remark}

\subsection{AKSZ $\sigma$-Models}
\label{section.AKSZSigmaModels}

We now consider, in definition \ref{TheAKSZAction} below,
for any symplectic dg-manifold $(X,\omega)$ a functional
$S_{\mathrm{AKSZ}}$ on spaces of maps $\mathfrak{T}\Sigma \to X$
of smooth graded manifolds,  and specialize this to the explicit formula
(\ref{AKSZformula}) in the special case the target manifold is endowed with global
Darboux coordinates. While only this particular situation is referred to in
the remainder of the article,
we begin by
indicating informally the original motivation of $S_{\mathrm{AKSZ}}$. The
reader uncomfortable with these somewhat vague considerations
can take formula (\ref{AKSZformula}) as a definition and then skip to the
next section.
\par
Generally, a \emph{$\sigma$-model field theory} is, roughly, one 
\begin{enumerate}
\item whose fields over a space $\Sigma$ are 
maps $\phi : \Sigma \to X$ 
to some space $X$;
\item
 whose action functional is, apart from a kinetic term,
 the transgression of some kind of cocycle  on $X$
 to the mapping space $\mathrm{Map}(\Sigma,X)$.
\end{enumerate}
Here the terms ``space'', ``maps'' and ``cocycles''
are to be made precise in a suitable context.
One says that $\Sigma$ is the \emph{worldvolume}, 
$X$ is the \emph{target space} and the cocycle is
the \emph{background gauge field}.

For instance, an ordinary charged particle 
(such as an electron) is described by a $\sigma$-model 
where $\Sigma = (0,t) \subset \mathbb{R}$ is the abstract
\emph{worldline}, where $X$ is a (pseudo-)Riemannian smooth manifold
(for instance our spacetime), and where the background cocycle
is a line bundle with connection on $X$ (a degree-2 cocycle
in ordinary differential cohomology of $X$, representing a
background \emph{electromagnetic field}). Up to a kinetic term, 
the action functional is the holonomy of the connection over 
a given curve $\phi : \Sigma \to X$. 
A textbook discussion of these
standard kinds of $\sigma$-models is, for instance, in \cite{DeligneMorgan}.

The $\sigma$-models which we consider here are 
\emph{higher} generalizations of this example,
where the background gauge field is a cocycle of higher degree
(a higher bundle with connection) and where the worldvolume is
accordingly higher dimensional. In addition, $X$ is allowed to be 
not just a manifold, but an approximation to a 
\emph{higher orbifold} (a smooth $\infty$-groupoid).

More precisely, here we take the category of spaces
to be $\mathrm{SmoothDgMfd}$ from def. \ref{DGManifolds}.
We take target space to be a symplectic dg-manifold $(X,\omega)$
and the worldvolume to be the shifted tangent bundle 
$\mathfrak{T}\Sigma$ of a compact smooth manifold $\Sigma$.
Following \cite{AKSZ}, one may imagine that we can form a 
smooth $\mathbb{Z}$-graded mapping space
$\mathrm{Maps}(\mathfrak{T}\Sigma,X)$ 
of smooth graded manifolds. On this space the canonical
vector fields $v_\Sigma$ and $v_X$ naturally have commuting 
actions from the left
and from the right, respectively,
so that their sum $v_\Sigma + v_X$ 
equips $\mathrm{Maps}(\mathfrak{T}\Sigma,X)$ itself with the structure of 
a differential graded smooth manifold.

Next we take the ``cocycle'' on $X$ (to be made precise
in the next section) to be 
the Hamiltonian $\pi$ (def. \ref{Hamiltonians}) 
of $v_X$ with respect to the symplectic
structure $\omega$, according to def. \ref{SymplecticDgManifold}.
One wants to assume that there is a kind of 
Riemannian structure on $\mathfrak{T}\Sigma$
that allows to form the transgression
$$
  \int_{\mathfrak{T}\Sigma} \mathrm{ev}^* \omega
  :=
  p_! \mathrm{ev}^* \omega
$$
by pull-push through the canonical correspondence
$$
  \xymatrix{
     \mathrm{Maps}(\mathfrak{T}\Sigma,X)
     \ar@{<-}[r]^{p}
     &
     \mathrm{Maps}(\mathfrak{T}\Sigma,X) \times \mathfrak{T}\Sigma
     \ar[r]^<<<<{\mathrm{ev}}
     &
     X
  }
  \,.
$$
When one succeeds in making this precise, one expects to 
find that $\int_{\mathfrak{T}\Sigma} \mathrm{ev}^* \omega$ is in 
turn a symplectic structure on the mapping space. 

This implies that the vector field
$v_\Sigma + v_X$ on mapping space has a Hamiltonian 
$$
  \mathbf{S} \in C^\infty(\mathrm{Maps}(\mathfrak{T}\Sigma,X))
  \,,\;\;\mbox{s.t.}\;\;
  \mathbf{d}\mathbf{S} = \iota_{v^{}_\Sigma + v^{}_X} 
  \int_{\mathfrak{T}\Sigma} \mathrm{ev}^* \omega
  \,.
$$
\begin{definition}
 \label{TheAKSZAction}
The grade-0 component
$$
  S_{\mathrm{AKSZ}}
  :=
  \mathbf{S}|_{\mathrm{Maps}(\mathfrak{T}\Sigma,X)_0}
$$ 
of the Hamiltonian $ \mathbf{S}$
constitutes a functional on the space
of morphisms of graded manifolds $\phi : \mathfrak{T}\Sigma \to X$. This is
the \emph{AKSZ action functional} defining the AKSZ $\sigma$-model
with target space $X$ and background field/cocycle $\omega$.
\end{definition}
In \cite{AKSZ}, this procedure is indicated only somewhat vaguely. 
The focus of attention there is on a discussion, from this perspective,
of the action functionals
of the 2-dimensional $\sigma$-models called the 
\emph{A-model} and the \emph{B-model}.
In \cite{RoytenbergAKSZ} a more detailed discussion of
the general construction is given, including an explicit formula
for $\mathbf{S}$, and
hence for $S_{\mathrm{AKSZ}}$, in the case that $X$ admits global Darboux coordinates. 
That formula is the following:
for $(X,\omega)$ a symplectic dg-manifold of grade $n$
with global Darboux coordinates $\{x^a\}$, 
$\Sigma$ a smooth compact manifold of dimension $(n+1)$
and $k \in \mathbb{R}$, the
\emph{AKSZ action functional}
$$
  S_{\mathrm{AKSZ}} : 
   \mathrm{SmoothGrMfd}(\mathfrak{T}\Sigma, X)
  \to
  \mathbb{R}
$$
is
\begin{equation}\label{AKSZformula}
  S_{\mathrm{AKSZ}} 
   :
  \phi
  \mapsto
  \int_{\Sigma} 
  \left(
    \frac{1}{2}\omega_{ab} \phi^a \wedge d_{\mathrm{dR}}\phi^b
    -
    \phi^* \pi
  \right)
  \,,
\end{equation}
where $\pi$ is the Hamiltonian for $v_X$ with respect to $\omega$
and where on the
right we are interpreting fields as forms on $\Sigma$ according to
prop. \ref{FormsFromGrManifoldMaps}.
\par
This formula hence defines an infinite class of $\sigma$-models
depending on the target space structure $(X, \omega)$. 
(One can also consider arbitrary relative factors between the first and
the second term, but below we shall find that the above choice is
singled out).
In \cite{AKSZ}, it was
already noticed that ordinary Chern-Simons theory is a special
case of this for $\omega$ of grade 2,  as is the Poisson $\sigma$-model
for $\omega$ of grade 1 (and hence, as shown there, also the A-model
and the B-model). 
The main example in \cite{RoytenbergAKSZ} spells out the
general case for $\omega$ of grade 2, which is called the
\emph{Courant $\sigma$-model} there. 
(We review and re-derive
all these examples in detail in \ref{section.examples} below.)

One nice aspect of this construction is that it follows immediately
that the full Hamiltonian $\mathbf{S}$ on the mapping space satisfies
$\{\mathbf{S}, \mathbf{S}\} = 0$. Moreover, using the standard formula
for the internal hom of chain complexes, one finds that the cohomology
of $(\mathrm{Maps}(\mathfrak{T}\Sigma,X), v_\Sigma + v_X)$ in degree 0
is the space of functions on those fields that satisfy the Euler-Lagrange
equations of $S_{\mathrm{AKSZ}}$. Taken together, these facts imply that 
$\mathbf{S}$ is a solution of the ``master equation''
of a BV-BRST complex for the quantum field theory
defined by $S_{\mathrm{AKSZ}}$. This is a crucial ingredient for the 
quantization of the model, and this is what the AKSZ construction
is mostly used for in the literature (for instance \cite{CattaneoFelder01}).

Here we want to focus on another nice aspect of the AKSZ-construction:
it hints at a deeper reason for \emph{why} the $\sigma$-models of 
this type are special. It is indeed one of the very few
proposals for what a general abstract mechanism might be that picks 
out among the vast space of all possible local action functionals 
those that seem to be of relevance ``in nature''. 

We now proceed to show that the class of action functionals
$S_{\mathrm{AKSZ}}$ are precisely those that higher Chern-Weil theory
canonically associates to target data $(X,\omega)$. 
Since higher Chern-Weil theory in turn is canonically given 
on very general abstract grounds \cite{survey}, this in 
a sense amounts to a derivation of $S_{\mathrm{AKSZ}}$ from 
``first principles'', and it shows that a wealth of very general 
theory applies to these systems. More on this will be discussed
elsewhere, some indication are in Section \ref{section.CS_theory}. 
Here we shall focus on a concrete computation 
exhibiting $S_{\mathrm{AKSZ}}$ as the image of the higher 
Chern-Weil homomorphism.

\section{Chern-Weil theory on $L_\infty$-algebroids}
\label{section.cs_elements}

We now discuss the $\infty$-Lie theoretic concepts
in terms of which we shall re-express the AKSZ $\sigma$-model
below in \ref{section.AKSZ_theory}.

\subsection{General $L_\infty$-algebroids}

We survey some basics of $\infty$-Lie theory that we need
later on. The explicit $L_\infty$-algebraic constructions are from
\cite{SSSI}, a more encompassing discussion is in \cite{survey}.

The following definition essentially repeats def. \ref{DGManifolds}
with different terminology. While this may look like a 
redundancy, it is useful to instead regard it as the 
beginning of a useful dictionary between higher Lie theory
and dg-geometry. The examples to follow will illustrate this.
\begin{definition}
  \label{LInfinityAlgebroids}
  The category of \emph{$L_\infty$-algebroids} is equivalent
  to that of smooth dg-manifolds from def. \ref{DGManifolds}:
  $$
    L_\infty \mathrm{Algd}
    \simeq
    \mathrm{SmoothDgMfd}
    \hookrightarrow
    \mathrm{cdgAlg}_{\mathbb{R}}^{\mathrm{op}}
    \,.
  $$
  For $\mathfrak{a} \in L_\infty\mathrm{Algd}$ we write
  $\mathrm{CE}(\mathfrak{a}) \in \mathrm{cdgAlg}_{\mathbb{R}}$
  for the corresponding dg-algebra and call it the
  \emph{Chevalley-Eilenberg algebra} of $\mathfrak{a}$.
  
  If the graded algebra underlying $\mathrm{CE}(\mathfrak{a})$ 
  has generators of grade at most $n$, we say that $\mathfrak{a}$
  is a \emph{Lie $n$-algebroid}.
\end{definition}
\noindent{\bf Examples.}
\begin{itemize}
 \item
 Any (degreewise finite-dimensional) 
 $L_\infty$-algebra (and so, in particular, any Lie algebra) $\mathfrak{g}$ can be seen as a (canonically pointed) 
$L_\infty$-algebroid $b\mathfrak{g}$ over the point:
for $D : \vee^\bullet \mathfrak{g} \to \vee^\bullet \mathfrak{g}$
the nilpotent derivation on the free graded coalgebra over $\mathfrak{g}$ which
defines the $k$-ary brackts on $\mathfrak{g}$ by
$$
  [-,-, \cdots,-]_k^{} := D|_{\vee^k \mathfrak{g}} : \vee^k \mathfrak{g} \to \mathfrak{g}
  \,,
$$
we have
$$
  \mathrm{CE}(b \mathfrak{g})
   :=
  (\wedge^\bullet \mathfrak{g}^*, d := D^*)
  \,.
$$
One directly finds that $L_\infty$-algebroid morphims $b\mathfrak{g} \to b\mathfrak{h}$ are precisely $L_\infty$-algebra morphism $\mathfrak{g}\to \mathfrak{h}$. This means that there is a full and faithful inclusion
$$
  b : L_\infty \mathrm{Alg} \hookrightarrow L_\infty \mathrm{Algd}
$$
of the traditional category of $L_\infty$-algebras into that
of $L_\infty$-algebroids.

We refer to $b\mathfrak{g}$ as the \emph{delooping} of $\mathfrak{g}$. This notation is the infinitesimal analog of the notation $\mathbf{B}G$ for the one-object 
Lie groupoid corresponding to a Lie group $G$: the loop space object 
$\Omega\mathbf{B}G$ is equivalent to $G$, hence the name ``delooping'' given to $\mathbf{B}G$. 

For $\mathfrak{g}$ a Lie algebra, the algebra $\mathrm{CE}(b\mathfrak{g})$ is the ordinary Chevalley-Eilenberg algebra of $\mathfrak{g}$.
\item
   For $n \in \mathbb{N}$ the delooping of the 
   \emph{line Lie $n$-algebra} is the $L_\infty$-algebroid
   $b^{n}\mathbb{R}$ defined by the fact that 
   $\mathrm{CE}(b^{n}\mathbb{R})$ is generated over 
   $\mathbb{R}$ from a single generator in degree $n$
   with vanishing differential.
 \item
  For $X$ a smooth manifold, the \emph{tangent Lie algebroid}
  $\mathfrak{a} = \mathfrak{T}X$ is defined by 
  $\mathrm{CE}(\mathfrak{T}X) = (\Omega^\bullet(X), d_{\mathrm{dR}})$ ;
 \item
   For $(X, \{-,-\})$ a Poisson manifold, the corresponding
  \emph{Poisson Lie algebroid} $\mathfrak{P}(X)$ is defined by
  $$
    \mathrm{CE}(\mathfrak{P}(X)) = 
     (\wedge^\bullet_{C^\infty(X)}
     \Gamma(T X),
       \{\pi,-\}
     )
     \,,
  $$
  where $\pi \in \wedge^2_{C^\infty(X)} \Gamma(T X)$ is the Poisson 
  tensor and the bracket means the canonical extension to the 
  tangent bundle: the Schouten bracket.
\end{itemize}
\begin{remark}
For $\mathfrak{a}$ an $L_\infty$-algebroid
and $\{x^i\}$ local coordinates on the corresponding
graded manifold, the vector field $v$ corresponding to 
the Chevellay-Eilenberg differential $d_{\mathrm{CE}(\mathfrak{a})}$ is
$$
v\bigr\vert_U=v^i\frac{\partial}{\partial x^i},
$$
with $v^i:=d_{\mathrm{CE}(\mathfrak{a})} x^i$.
\end{remark}
\begin{definition}
  \label{WeilAlgebraAbstract}
  For $\mathfrak{a}$ an $L_\infty$-algebroid,
  its \emph{Weil algebra} is that representative
  of the free \emph{smooth dg-algebra}, 
  remark \ref{SmoothDgAlgebras}, 
  on the underlying word-length-1 complex of $\mathfrak{a}$ that makes the
  canonical projection of complexes
  $$
    i^* : W(\mathfrak{a}) \to \mathrm{CE}(\mathfrak{a})
  $$
  into a dg-algebra homomorphism.
\end{definition}
\begin{proposition}
  \label{WeilAlgebraExplicit}
  Explicitly, the Weil algebra $\mathrm{W}(\mathfrak{a})$ has
  \begin{itemize}
    \item as underlying graded algebra the de Rham complex
      $\Omega^\bullet(\mathfrak{a})$
      from def. \ref{deRhamComplex}, applied to the 
      corresponding graded manifold; i.e., the differential graded manifold corresponding to $\mathrm{W}(\mathfrak{a})$ is the tangent Lie $\infty$-algebroid $\mathfrak{T}\mathfrak{a}$. This can be equivalently written as
$$
\mathrm{W}(\mathfrak{a})=\mathrm{CE}(\mathfrak{T}\mathfrak{a}).
$$
    \item
      as differential the sum
      $$
        d_{\mathrm{W}(\mathfrak{a})}
        =
        \mathbf{d} + \mathcal{L}_v
        \,,
      $$
     where $\mathbf{d}$ is the differential from def. \ref{deRhamComplex}
     and where $\mathcal{L}_v$ is the Lie derivative along the 
     vector field $v$ corresponding to the Chevalley-Eilenberg differential.
   \end{itemize}  
\end{proposition}
\noindent{\bf Remark.}
Therefore the Weil algebra $\mathrm{W}(\mathfrak{a})$ is a
\emph{twisted} de Rham complex on the graded smooth manifold
corresponding to $\mathfrak{a}$, where the twist is dictated
by the characterizing morphism $i^*$ from def. \ref{WeilAlgebraAbstract}.
In the abstract theory indicated below in 
\ref{section.CS_theory} this makes $\mathrm{W}(\mathfrak{a})$
part of the construction of a certain homotopical resolution of the 
Lie integration of $\mathfrak{a}$. This is the deeper
reason for the role played by the Weil algebra in higher Lie theory. But
for the present purpose the above explicit definition is sufficient.
\\
\begin{examples}
\begin{itemize}
  \item For $\mathfrak{g}$ a Lie algebra, the definition of
  $\mathrm{W}(b\mathfrak{g})$ reduces to the ordinary definition of the Weil algebra.
  \item
    For $\mathfrak{a} = \Sigma$ an ordinary smooth manifold, 
    $\mathrm{W}(\Sigma) = \Omega^\bullet(\Sigma)$.
  \item
    For $G$ a Lie group with Lie algebra $\mathfrak{g}$ 
    acting on a manifold $\Sigma$, write
    $\Sigma//\mathfrak{g}$ for the corresponding action Lie algebroid.
    Then $\mathrm{W}(\Sigma//\mathfrak{g})$ is the Cartan-Weil model 
    for $G$-equivariant de Rham cohomology on $\Sigma$.
  \item
    For $\mathfrak{a} = b^n \mathbb{R}$ the delooping of the 
    line Lie $n$-algebra, we have that $\mathrm{W}(b^n \mathbb{R})$
    is the free dg-algebra on a single generator $c$ in degree $n$:
    this is the graded algebra on two generators $c$ and $\gamma$, with $c$ in degree $n$ and $\gamma$ in degree $n+1$,
    equipped with a differential defined by 
    $
      d_{\mathrm{W}(b^n \mathbb{R})}
        : c \mapsto \gamma
        \,.
    $
\end{itemize}
\end{examples}

\subsection{Cocycles, invariant polynomials and Chern-Simons elements 
}
\label{section.CWMorphism}


The key technical notion for our main theorem is that of 
Chern-Simons elements witnessing trangression between
invariant polynomials and $L_\infty$-algebroid cocycles,
which is def. \ref{TransgressionAndCSElements} below. 
We show in Section \ref{section.CS_theory} how 
these notions are related to the $\infty$-Chern-Weil homomorphism for $\infty$-bundles with connections.

\begin{definition}
  Let $\mathfrak{a}$ be an $L_\infty$-algebroid.
  An \emph{$L_\infty$-cocycle} on $\mathfrak{a}$ is an element
  $\mu \in \mathrm{CE}(\mathfrak{a})$ which is closed.
\end{definition}
\begin{examples}
  \begin{itemize}
    \item For $\mathfrak{a} = b \mathfrak{g}$ the delooping of an 
    ordinary Lie algebra, $L_\infty$-cocycles on $\mathfrak{a}$
    are precisely traditional Lie algebra cocycles on $\mathfrak{g}$.
    \item
      For $n \in \mathbb{N}$ and $\mathfrak{a} = b^n \mathbb{R}$
      the delooping of the line Lie $n$-algebra, there is, up to scale,
      precisely one nontrivial $L_\infty$-cocycle on $\mathfrak{a}$, 
      which is in degree $n$.
  \end{itemize}
\end{examples}
\begin{definition}
  \label{InvariantPolynomial}  
  An \emph{invariant polynomial} on $\mathfrak{a}$
  is an element  $\langle - \rangle$ in $\mathrm{W}(\mathfrak{a})$ which is
  \begin{enumerate}
    \item closed: $d_{\mathrm{W}(\mathfrak{a})} \langle - \rangle = 0$;
    \item horizontal: an element of the subalgebra 
     generated by the shifted elements in the Weil algebra.
  \end{enumerate}
\end{definition}
\begin{examples}
  \label{ExamplesOfInvariantPolynomials}
  \begin{itemize}
     \item 
    For $\mathfrak{a} = b \mathfrak{g}$ the delooping of an 
    ordinary Lie algebra, one readily
     checks that the above definition reproduces the traditional
     definition of invariant polynomials.
     \item
       For $\mathfrak{a} = \Sigma$ a 0-Lie algebroid (a smooth manifold),
       an invariant polynomial is a closed differential form of 
       positive degree.
     \item
       For $n \in \mathbb{N}$ and $\mathfrak{a} = b^n \mathbb{R}$
       the delooping of the line Lie $n$-algebra, there is 
       a 1-dimensional vector space of invariant polynomials
       of degree $(n+1)$ and every other homogeneous invariant polynomial
       is a wedge power of these. In particular for even $n$
       all further invariant polynomials vanish.
  \end{itemize}
\end{examples}
%
\begin{definition}
  \label{TransgressionAndCSElements}
  For $\langle -\rangle\in \mathrm{W}(\mathfrak{a})$ an invariant
  polynomial on an $L_\infty$-algebroid $\mathfrak{a}$,
  we say a cocycle $\mu \in \mathrm{CE}(\mathfrak{a})$
  is \emph{in transgression} with $\langle -\rangle$ if
  there exists an element 
  $\mathrm{cs}$ in $\mathrm{W}(\mathfrak{a})$ such that
  \begin{enumerate}
    \item
       $d_{\mathrm{W}(\mathfrak{a})} \mathrm{cs} = \langle -\rangle$;
    \item
       $i^* \mathrm{cs} = \mu$.
  \end{enumerate}
  We say that $\mathrm{cs}$ is a
  \emph{transgression element} or 
  \emph{Chern-Simons element} witnessing this transgression.
\end{definition}

As we noticed above, if we look at an ordinary smooth manifold $\Sigma$ as an $L_\infty$-algebroid, then the Weil algebra of $\Sigma$ is the de Rham algebra $\Omega^\bullet(\Sigma)$. This motivates the following definition.
 \begin{definition}
  \label{definition.a-valued-differential-form}
  For $\mathfrak{a}$ an $L_\infty$-algebroid
  and $\Sigma$ a smooth manifold, we say a morphism
  $$
   A:\mathrm{W}(\mathfrak{a})\to \Omega^\bullet(\Sigma)  $$
  is a degree 1 \emph{$\mathfrak{a}$-valued differential form} on $\Sigma$.
  \end{definition}
  \begin{remark}
  The name ``degree 1 $\mathfrak{a}$-valued differential forms'' given to dgca morphisms $\mathrm{W}(\mathfrak{a})\to \Omega^\bullet(\Sigma)$ has the following origin: if $\mathfrak{g}$ is a Lie algebra, then the Weil algebra $\mathrm{W}(b\mathfrak{g})$ is the free differential graded commutative algebra generated by a shifted copy $\mathfrak{g}^*[-1]$ of the linear dual of $\mathfrak{g}$. Hence a dgca morphism $\mathrm{W}(b\mathfrak{g})\to \Omega^\bullet(\Sigma)$ is precisely the datum of a morphism of graded vector spaces $\mathfrak{g}^*[-1]\to \Omega^\bullet(\Sigma)$, i.e., an element of $\Omega^1(\Sigma;\mathfrak{g})$.
   
  \end{remark}
  We say that an $\mathfrak{a}$-valued differential form $A$ is \emph{flat} if the  morphism $A:\mathrm{W}(\mathfrak{a})\to \Omega^\bullet(\Sigma)$ factors through
  $i^* : \mathrm{W}(\mathfrak{a}) \to \mathrm{CE}(\mathfrak{a})$.
   The \emph{curvature} of $A$ is the induced morphism of 
  graded vector spaces given by the composite
 $$
    \xymatrix{
      \Omega^\bullet(\Sigma)
      \ar@{<-}[r]^{A}
      &
      \mathrm{W}(\mathfrak{a})
      \ar@{<-}[r]
      &
      \wedge^1 V[1]
      : 
      F_A
      \,,
    }
  $$
  where the morphism on the right is the inclusion of the linear subspace of the shifted generators into the Weil algebra. $A$ is flat precisely if 
$F_A = 0$.
\begin{remark}
  \label{ImagesOfdAnddWUnderForms}
   For $\{x^a\}$ a coordinate chart of $\mathfrak{a}$
  and
  $$
    A^a := A(x^a) \in \Omega^{\mathrm{deg}(x^a)}(\Sigma)
  $$
  the differential form assigned to the generator $x^a$ by 
  the $\mathfrak{a}$-valued form $A$, we have the 
  curvature components
  $$
    F_A^a = A(\mathbf{d}x^a) \in \Omega^{\mathrm{deg}(x^a)+1}(\Sigma)
    \,.
  $$
  Since $d_{\mathrm{W}}=d_{\mathrm{CE}}+\mathbf{d}$, this can be equivalently written as
  $$
    F_A^a = A(d_{\mathrm{W}}x^a-d_{\mathrm{CE}}x^a)
    \,,
  $$
  so the \emph{curvature} of $A$ precisely measures the ``lack of flatness'' of $A$.
  Also notice that, since $A$ is required to be 
  a dg-algebra homomorphism, we have
  $$
    A(d_{\mathrm{W}(\mathfrak{a})} x^a) = d_{\mathrm{dR}} A^a
    \,,
  $$
   so that
  $$
    A(d_{\mathrm{CE}(\mathfrak{a})} x^a) = d_{\mathrm{dR}} A^a- F_A^a
    \,.
  $$ 
   
\end{remark}

Assume now $A$ is a degree 1 $\mathfrak{a}$-valued differential form on the smooth manifold $\Sigma$, and that $\mathrm{cs}$ is a Chern-Simons element transgressing an invariant polynomial $\langle-\rangle$ of $\mathfrak{a}$ to some cocycle $\mu$. We can then consider the image $A(\mathrm{cs})$ of the Chern-Simons element $\mathrm{cs}$ in $\Omega^\bullet(\Sigma)$. Equivalently, we can look at $\mathrm{cs}$ as a map from degree 1 $\mathfrak{a}$-valued differential forms on $\Sigma$ to ordinary (real valued) differential forms on $\Sigma$.
\begin{definition}
\label{definition.chern-simons-form} 
In the notations above, we write 
$$
  \xymatrix{
    \Omega^\bullet(\Sigma)
     \ar@{<-}[r]^{A} &
    \mathrm{W}(\mathfrak{a})
     \ar@{<-}[r]^{\mathrm{cs}} &
    \mathrm{W}(b^{n+1}\mathbb{R})
  }
  :
  \mathrm{cs}(A)
$$ 
for the differential form associated by the Chern-Simons element $\mathrm{cs}$ to the degree 1 $\mathfrak{a}$-valued differential form $A$, 
and call this the \emph{Chern-Simons differential form} associated with $A$. 

Similarly, for $\langle -\rangle$ an invariant polynomial on 
$\mathfrak{a}$, we write $\langle F_A \rangle$
for the evaluation
$$
  \xymatrix{
    \Omega^\bullet_{\mathrm{closed}}(\Sigma)
     \ar@{<-}[r]^{A} &
    \mathrm{W}(\mathfrak{a})
     \ar@{<-}[r]^{\langle -\rangle} &
    \mathrm{inv}(b^{n+1}\mathbb{R})
  }
  :
  \langle F_A \rangle
  \,.
$$ 
We call this the \emph{curvature characteristic form}
of $A$ with respect to $\langle-\rangle$.
\end{definition}

\subsection{Symplectic Lie $n$-algebroids}

We now consider $L_\infty$-algebroids that are equipped with 
certain natural extra structure (symplectic structure) and show how this
extra structure canonically induces an invariant polynomial
and hence by observation \ref{LagrangianFromCW} a
$\sigma$-model field theory. In the next section we 
demonstrate that the field theories arising this way
are precisely the AKSZ $\sigma$-models.

\begin{definition}
 \label{SymplecticLooAlgebroid}
A \emph{symplectic Lie $n$-algebroid}
\index{$L_\infty$-algebroid!symplectic Lie $n$-algebroid} 
\index{symplectic Lie $n$-algebroid}
$(\mathfrak{P}, \omega)$ is a 
Lie $n$-algebroid $\mathfrak{P}$ equipped with a quadratic non-degenerate invariant 
polynomial $\omega \in W(\mathfrak{P})$ of degree $n+2$.
\end{definition}
This means that 
\begin{itemize}
\item on each chart $U \to X$ of the base manifold $X$ of 
$\mathfrak{P}$, there is a 
basis $\{x^a\}$ for $\mathrm{CE}(\mathfrak{a}|_U)$ such that
$$
  \omega =  \frac{1}{2}\mathbf{d}x^a \,\omega_{a b}\wedge \mathbf{d}x^b
$$
with $\{\omega_{a b} \in \mathbb{R} \hookrightarrow C^\infty(X)\}$ 
and $\mathrm{deg}(x^a) + \mathrm{deg}(x^b) = n$;
\item the coefficient matrix $\{\omega_{a b}\}$ has an inverse;
\item we have
$$
  d_{\mathrm{W}(\mathfrak{P})} \omega
  = 
  d_{\mathrm{CE}(\mathfrak{P})} \omega 
  +
  \mathbf{d} \omega = 0
  \,.
$$
\end{itemize}
The following observation essentially goes back to \cite{Severa}
and \cite{RoytenbergCourant}.
\begin{proposition}
  \label{SymplDgSpaceAsLAlgd}
  There is a full and faithful embedding of symplectic dg-manifolds
  of grade $n$
  into symplectic Lie $n$-algebroids.
\end{proposition}
\proof
  The dg-manifold itself is identified with an $L_\infty$-algebroid
  by def. \ref{LInfinityAlgebroids}.
  For $\omega \in \Omega^2(X)$ a symplectic form, the conditions
  $\mathbf{d} \omega = 0$ and $\mathcal{L}_v \omega = 0$
  imply $(\mathbf{d}+ \mathcal{L}_v)\omega = 0$ 
  and hence that under the identification 
  $\Omega^\bullet(X) \simeq \mathrm{W}(\mathfrak{a})$ 
  this is an invariant polynomial on $\mathfrak{a}$.
  
  It remains to observe that the $L_\infty$-algebroid $\mathfrak{a}$
  is in fact a Lie $n$-algebroid. This is implied by the 
  fact that $\omega$ is of grade $n$ and non-degenerate: 
  the former condition implies that it has no components in elements
  of grade
  $> n$ and the latter then implies that all such elements vanish.
\endofproof
The following characterization may be taken as a definition 
of Poisson Lie algebroids and Courant Lie 2-algebroids.
\begin{proposition}
  \label{SymplecticPoissonCourant}
  Symplectic Lie $n$-algebroids are equivalently:
  \begin{itemize}
    \item for $n = 0$: ordinary symplectic manifolds;
    \item for $n = 1$: Poisson Lie algebroids;
    \item for $n = 2$: Courant Lie 2-algebroids.
  \end{itemize}
\end{proposition}
See \cite{RoytenbergCourant, Severa} for more discussion.
\begin{proposition}
 \label{CanonicalCocycleOfSymplecticLieAlgebroid}
Let $(\mathfrak{P},\omega)$ be a symplectic Lie $n$-algebroid 
for positive $n$ in the image of the embedding of 
proposition \ref{SymplDgSpaceAsLAlgd}.
Then it carries the canonical $L_\infty$-algebroid cocycle 
 $$
   \pi := \frac{1}{n+1} \iota_\epsilon \iota_v \omega
   \in 
   \mathrm{CE}(\mathfrak{P})
 $$ 
 which moreover is the Hamiltonian, 
 according to definition \ref{Hamiltonians}, of $d_{\mathrm{CE}(\mathfrak{P})}$.
\end{proposition} 
 \proof Since $\mathbf{d}\omega=\mathcal{L}_v\omega=0$, we have
 $$
 \begin{aligned} 
\mathbf{d} \iota_\epsilon \iota_v \omega
&= \mathbf{d} \iota_v \iota_\epsilon  \omega\\
&=(\iota_v\mathbf{d}-\mathcal{L}_v)\iota_\epsilon  \omega\\
&=\iota_v\mathcal{L}_\epsilon\omega-[\mathcal{L}_v,\iota_\epsilon]\omega\\
&=n\iota_v\omega-\iota_{[v,\epsilon]}\omega\\
&=(n+1)\iota_v\omega,
\end{aligned}
$$
where Cartan's formula $[\mathcal{L}_v,\iota_\epsilon]=\iota_{[v,\epsilon]}$ and the identity $[v,\epsilon]=-[\epsilon,v]=-v$ have been used. Therefore $\pi:=\frac{1}{n+1}\iota_\epsilon \iota_v \omega$ satisfies the defining equation $\mathbf{d}\pi=\iota_v\omega$
from definition \ref{Hamiltonians}.
\endofproof

\begin{remark}
  \label{remark.local_hamiltonian}
On a local chart 
with coordinates $\{x^a\}$ 
we have
$$
    \pi\bigr\vert_U 
    =
    \frac{1}{n+1}\omega_{ab}\;\deg(x^a) x^a\, \wedge v^b
  \,.
$$
\end{remark}
Our central observation now is the following.
\begin{proposition}
 \label{TheCSElement}
The cocycle $\frac{1}{n} \pi$ 
from prop. \ref{CanonicalCocycleOfSymplecticLieAlgebroid}
is in transgression with 
the invariant polynomial $\omega$.
A Chern-Simons element witnessing the transgression 
according to def. \ref{TransgressionAndCSElements} is
$$
  \mathrm{cs} = \frac{1}{n}\left(\iota_\epsilon \omega + \pi\right)
  \,.
$$
\end{proposition}
\proof
It is clear that $i^* \mathrm{cs} = \frac{1}{n}\pi$. So it remains to 
check that $d_{\mathrm{W}(\mathfrak{P})} \mathrm{cs} = \omega$.
As in the proof of proposition \ref{CanonicalCocycleOfSymplecticLieAlgebroid}, we use $\mathbf{d}\omega=\mathcal{L}_v\omega=0$ and Cartan's identity $ [\mathcal{L}_v, \iota_\epsilon]
  =
  \iota_{[v,\epsilon]}
  =
  - \iota_{v}
$.
By these, the first summand in
$d_{\mathrm{W}(\mathfrak{P})} ( \iota_{\epsilon} \omega + \pi )$ is
$$
  \begin{aligned}
    d_{\mathrm{W}(\mathfrak{P})} \iota_{\epsilon} \omega
     & = 
     (
       \mathbf{d}
       + \mathcal{L}_v
     )
     \iota_\epsilon \omega
     \\
     &=
     [\mathbf{d}
       +\mathcal{L}_v,\iota_\epsilon]\omega\\
       &= n\omega - \iota_v \omega
     \\
     & = n \omega - \mathbf{d}\pi
  \end{aligned}
  \,.
$$
The second summand is simply 
$$
  d_{\mathrm{W}(\mathfrak{P})} \pi =  \mathbf{d}\pi
$$
since $\pi$ is a cocycle. 
\endofproof
\begin{remark}
 \label{remark.local_cs}
In a coordinate patch $\{x^a\}$ the Chern-Simons element is
$$
 \mathrm{cs}\bigr\vert_U
 =
\frac{1}{n}
   \left(
     \omega_{a b} \,\mathrm{deg}(x^a) x^a\, \wedge \mathbf{d}x^b + \pi
   \right)
   \,.
 $$
 In this formula one can substitute $\mathbf{d} = d_{\mathrm{W}}- d_{\mathrm{CE}}$, and this kind of substitution will be 
crucial for the proof our main statement in proposition \ref{TheAKSZActionFromCS} below.
Since $d_{\mathrm{CE}} x^i=v^i$ and using
remark \ref{remark.local_hamiltonian} we find
$$
   \sum_a \omega_{a b} \mathrm{deg}(x^a) x^a \wedge d_{\mathrm{CE}} x^b
   = 
  (n+1) \pi
  \,,
$$
and hence
$$
  \mathrm{cs}\bigr\vert_U
   =
  \frac{1}{n}
  \left(
    \mathrm{deg}(x^a) \,
    \omega_{a b} x^a \wedge d_{\mathrm{W}(\mathfrak{P})} x^b
    -
    n \pi
  \right)
  \,.
$$
\end{remark}
In the following section we show that this
transgression element $\mathrm{cs}$ \emph{is} the 
AKSZ-Lagrangian.

\section{The AKSZ action as a Chern-Simons functional}
\label{section.AKSZ_theory}
We now show how an $L_\infty$-algebroid $\mathfrak{a}$ endowed with a triple $(\pi,\mathrm{cs},\omega)$ consisting of a Chern-Simons element transgressing an invariant polynomial $\omega$ to a cocycle $\pi$ defines an AKSZ-type $\sigma$-model action. The starting point is to take as target space the tangent Lie $\infty$-algebroid $\mathfrak{T}\mathfrak{a}$, i.e., to consider as \emph{space of fields} of the theory the space of maps $\mathrm{Maps}(\mathfrak{T}\Sigma,\mathfrak{T}\mathfrak{a})$ from the worldsheet $\Sigma$ to $\mathfrak{T}\mathfrak{a}$. Dually, this is the space of morphisms of dgcas from
$\mathrm{W}(\mathfrak{a})$ to $\Omega^\bullet(\Sigma)$, i.e., the space of degree 1 $\mathfrak{a}$-valued differential forms on $\Sigma$ from definition \ref{definition.a-valued-differential-form}.

\begin{remark}
As we noticed in the introduction, in the context of the AKSZ $\sigma$-model a degree 1 $\mathfrak{a}$-valued differential form on $\Sigma$ should be thought of as the datum of a (notrivial) $\mathfrak{a}$-valued connection on a trivial principal $\infty$-bundle on $\Sigma$. We will come back to this point of view in Section \ref{section.CS_theory}.
\end{remark}
Now that we have defined the space of fields, we have to define the action. We have seen in definition \ref{definition.chern-simons-form} that a degree 1 $\mathfrak{a}$-valued differential form $A$ on $\Sigma$ maps the Chern-Simons element $\mathrm{cs}\in \mathrm{W}(\mathfrak{a})$ to a differential form $\mathrm{cs}(A)$ on $\Sigma$. Integrating this differential form on $\Sigma$ will therefore give an AKSZ-type action which, as we will see in Section \ref{section.CS_theory}, is naturally interpreted as an higher Chern-Simons action functional:
$$
\begin{aligned}
\mathrm{Maps}(\mathfrak{T}\Sigma,\mathfrak{T}\mathfrak{a}) &\to \mathbb{R}\\
A&\mapsto \int_\Sigma \mathrm{cs}(A).
\end{aligned}
$$

Theorem \ref{AKSZIsCW} then reduces to showing that, when $\{\mathfrak{a}, (\pi,\mathrm{cs},\omega)\}$ is the set of $L_\infty$-algebroid data arising from a symplectic Lie $n$-algebroid $(\mathfrak{P}, \omega)$, the AKSZ-type action dscribed above is precisely the AKSZ action for $(\mathfrak{P}, \omega)$. More precisely, this is stated as follows.
\begin{proposition}
  \label{TheAKSZActionFromCS}
  For $(\mathfrak{P}, \omega)$ a symplectic Lie $n$-algebroid 
  coming by proposition \ref{SymplDgSpaceAsLAlgd}
  from a symplectic dg-manifold of positive grade $n$ with 
  global Darboux chart, the action functional induced by 
  the canonical Chern-Simons element 
  $$
     \mathrm{cs} \in \mathrm{W}(\mathfrak{P})
  $$ 
  from proposition \ref{TheCSElement}
  is the AKSZ action from formula \eqref{AKSZformula}:
  $$
    \int_\Sigma \mathrm{cs}
     =
    \int_\Sigma L_{\mathrm{AKSZ}}
    \,.
  $$
  In fact the two Lagrangians differ at most by an exact term
  $$
    \mathrm{cs} \sim L_{\mathrm{AKSZ}}
    \,.
  $$
\end{proposition}
\proof
We have seen in remark \ref{remark.local_cs} that in Darboux coordinates $\{x^a\}$ where
$$
  \omega = \frac{1}{2}\omega_{a b} \mathbf{d}x^a \wedge \mathbf{d}x^b
$$ 
the Chern-Simons element from proposition \ref{TheCSElement} is
given by
$$
  \mathrm{cs} 
    = \frac{1}{n}
  \left(
    \mathrm{deg}(x^a) \,
    \omega_{a b} x^a \wedge d_{\mathrm{W}(\mathfrak{P})} x^b
    -
    n \pi
  \right)\,.
  $$
This means that for $\Sigma$ an $(n+1)$-dimensional manifold and
$$
  \Omega^\bullet(\Sigma) 
    \leftarrow 
  \mathrm{W}(\mathfrak{P})
   :
  \phi
$$
a (degree 1) $\mathfrak{P}$-valued differential form on $\Sigma$
we have 
$$
  \begin{aligned}
      \int_\Sigma \mathrm{cs}(\phi)
      &= 
      \frac{1}{n}
      \int_\Sigma
      \left(
       \sum_{a,b} 
       \mathrm{deg}(x^a)\,\omega_{a b} \phi^a \wedge  d_{\mathrm{dR}} \phi^b 
       -
       n \pi(\phi) 
     \right)
  \end{aligned}
  \,,
$$
where we used $\phi(d_{\mathrm{W}(\mathfrak{P})} x^b)=d_{\mathrm{dR}} \phi^b$, as in  remark \ref{ImagesOfdAnddWUnderForms}.
Here the asymmetry in the coefficients of the first term is only
apparent. Using integration by parts on a closed $\Sigma$ 
we have
$$
  \begin{aligned}
    \int_\Sigma
    \sum_{a,b} \mathrm{deg}(x^a)\,\omega_{a b} \phi^a \wedge  d_{\mathrm{dR}} \phi^b 
    & =
    \int_\Sigma
     \sum_{a,b} (-1)^{1+\mathrm{deg}(x^a)}\mathrm{deg}(x^a)\,\omega_{a b} (d_{\mathrm{dR}} \phi^a) \wedge   \phi^b 
    \\
    & =
    \int_\Sigma
    \sum_{a,b} (-1)^{(1+\mathrm{deg}(x^a))(1+\mathrm{deg}(x^b))}\mathrm{deg}(x^a)\,\omega_{a b} 
  \phi^b \wedge (d_{\mathrm{dR}} \phi^a)     
    \\
    & =
    \int_\Sigma \sum_{a,b} \mathrm{deg}(x^b)\,\omega_{a b} 
      \phi^a \wedge (d_{\mathrm{dR}} \phi^b)     
  \end{aligned}
  \,,
$$
where in the last step we switched the indices on $\omega$ and used that
$\omega_{ab} = (-1)^{(1+\mathrm{deg}(x^a))(1+\mathrm{deg}(x^b))} \omega_{b a}$.
Therefore
$$
  \begin{aligned}
    \int_\Sigma
    \sum_{a,b} \mathrm{deg}(x^a)\,\omega_{a b} \phi^a \wedge  d_{\mathrm{dR}} \phi^b 
    & =
    \frac{1}{2}
    \int_\Sigma
    \sum_{a,b} \mathrm{deg}(x^a)\,\omega_{a b} \phi^a \wedge  d_{\mathrm{dR}} \phi^b 
    +
    \frac{1}{2}
    \int_\Sigma
    \sum_{a,b} \mathrm{deg}(x^b)\,\omega_{a b} \phi^a \wedge  d_{\mathrm{dR}} \phi^b 
    \\
    &=
    \frac{n}{2}
    \int_\Sigma
      \omega_{a b} \phi^a \wedge  d_{\mathrm{dR}} \phi^b 
    \,.
  \end{aligned}
  \,.
$$
Using this in the above expression for the action yields 
$$
  \int_\Sigma \mathrm{cs}(\phi)
  = 
  \int_\Sigma
   \left(
   \frac{1}{2}\omega_{ab} \phi^a \wedge d_{\mathrm{dR}} \phi^b
   -
   \pi(\phi)
   \right)
   \,,
$$
which is formula  \eqref{AKSZformula} for the action functional.
\endofproof

\begin{corollary}\label{KS4.4}
In the hypothesis of Proposition 4.2, if $N$ is an $(n+2)$-dimensional compact oriented
manifold with $\partial N=\Sigma$, then $$ \int_\Sigma L_{\mathrm{AKSZ}} =
\int_N \omega(F_\phi), $$ where $\omega(F_\phi)$ is the symplectic form
$\omega$, seen as an invariant polynomial, evaluated on the curvature of
$\phi:\mathrm{W}(\mathfrak{P})\to \Omega^\bullet(\Sigma)$.
\end{corollary}

\subsection{Examples}
\label{section.examples}

We unwind the general statement of proposition \ref{TheAKSZActionFromCS}
and its ingredients
in the central examples of interest, from proposition \ref{SymplecticPoissonCourant}:
the ordinary Chern-Simons action functional, the Poisson $\sigma$-model Lagrangian, the Courant $\sigma$-model Lagrangian, and the higher abelian Chern-Simons 
functional.
(The ordinary Chern-Simons model is the special 
case of the Courant $\sigma$-model for $\mathfrak{P}$ having as
base manifold the point. But since it is the archetype of
all models considered here, it deserves its own discussion.)

By the very content of proposition \ref{TheAKSZActionFromCS}
there are no surprises here and the following essentially
amounts to a review of the standard formulas for these 
examples. But it may be helpful to see our general 
$\infty$-Lie theoretic derivation of these formulas spelled out 
in concrete cases, if only to carefully track the various signs
and prefactors.

\subsubsection{Ordinary Chern-Simons theory}
\label{OrdinaryChernSimonsTheoryAsAKSZ}

Let $\mathfrak{P} = b\mathfrak{g}$
be a semisimple Lie algebra regarded as an $L_\infty$-algebroid 
with base space the point
and let $\omega := \langle -,-\rangle\in \mathrm{W}(b\mathfrak{g})$ be 
its Killing form invariant polynomial. Then 
$(b \mathfrak{g}, \langle -,-\rangle)$ is a symplectic Lie 2-algebroid.

For $\{t^a\}$ a dual basis
for $\mathfrak{g}$, being generators of grade 1 in 
$\mathrm{W}(\mathfrak{g})$ we have
$$
  d_{\mathrm{W}} t^a = - \frac{1}{2}C^a{}_{b c} t^a \wedge t^b + 
   \mathbf{d}t^a
$$
where $C^a{}_{b c} := t^a([t_b,t_c])$ and
$$
  \omega = \frac{1}{2} P_{a b} \mathbf{d}t^a \wedge \mathbf{d}t^b
  \,,
$$
where $P_{ab} := \langle t_a, t_b \rangle$.
The Hamiltonian cocycle $\pi$ from prop. \ref{CanonicalCocycleOfSymplecticLieAlgebroid}
is
$$
  \begin{aligned}
    \pi
     & =
     \frac{1}{2+1}\iota_v \iota_\epsilon \omega
     \\
     &= \frac{1}{3} \iota_v P_{ab} t^a \wedge \mathbf{d}t^b
     \\
     & =-\frac{1}{6}P_{ab}C^b_{cd}t^a\wedge t^c\wedge t^d
     \\
     & =: -\frac{1}{6}C_{abc}t^a\wedge t^b\wedge t^c.
  \end{aligned}
$$
Therefore the Chern-Simons element from 
prop. \ref{TheCSElement} is found to be
$$
\begin{aligned}
 \mathrm{cs}
  &=
  \frac{1}{2}\left(P_{ab}t^a\wedge\mathbf{d}t^b
  -
  \frac{1}{6}C_{abc}t^a\wedge t^b\wedge t^c\right)
   \\
   &=
   \frac{1}{2}\left(P_{ab}t^a\wedge d_\mathrm{W}t^b
   +
   \frac{1}{3}C_{abc}t^a\wedge t^b\wedge t^c\right).
\end{aligned}
$$
This is indeed, up to an overall factor $1/2$, the familiar standard choice of Chern-Simons element on a Lie algebra. To see this more explicitly,
notice that evaluated on a
$\mathfrak{g}$-valued connection form
$$
  \Omega^\bullet(\Sigma) \leftarrow \mathrm{W}(b\mathfrak{g}) : A
$$
this is
$$
  2 \mathrm{cs}(A) 
   = 
  \langle A \wedge F_A\rangle 
   - 
  \frac{1}{6}\langle A \wedge [A, A]\rangle
   = 
  \langle A \wedge d_{dR}A\rangle 
   + 
  \frac{1}{3}\langle A \wedge [A, A]\rangle
  \,.
$$
If $\mathfrak{g}$ is a matrix Lie algebra then the Killing form is 
proportional to the trace of the matrix product: $\langle t_a,t_b\rangle = \mathrm{tr}(t_a t_b)$. In this case we have
$$
  \begin{aligned}
    \langle A \wedge [A, A]\rangle
    &=
    A^a  \wedge A^b \wedge A^c \,\mathrm{tr}(t_a (t_b t_c - t_c t_b))
    \\
    &=
    2 A^a \wedge A^b \wedge A^c \,\mathrm{tr}(t_a t_b t_c )
    \\
    &= 
    2 \,\mathrm{tr}(A \wedge A \wedge A)
  \end{aligned}
$$
and hence
$$
  2 \mathrm{cs}(A) 
   = \mathrm{tr}\left(
A \wedge F_A
   - 
  \frac{1}{3} A \wedge A \wedge A\right)
   = 
\mathrm{tr}\left( A \wedge d_{dR}A 
   + 
  \frac{2}{3} A \wedge  A \wedge A\right)
  \,.
$$

\subsubsection{Poisson $\sigma$-model}
\label{section.PoissonSigmaModel}

Let  $(M, \{-,-\})$ be a Poisson manifold and 
let $\mathfrak{P}$ be the corresponding 
Poisson Lie algebroid. This is a symplectic Lie 1-algebroid.
Over a chart for the shifted cotangent bundle $T^*[-1]X$
with coordinates $\{x^i\}$ of degree 0
and  $\{\partial_i\}$ of degree 1, respectively, we have
$$
  d_{\mathrm{W}} x^i = -\pi^{i j}\mathbf{\partial}_j + \mathbf{d}x^i;
  \qquad
 d_{\mathrm{W}} \partial_i = 
 \frac{1}{2} \frac{\partial \pi^{j k}}{\partial x^i}
 \partial_j \wedge \partial_k + \mathbf{d}\partial_i
 \,,
$$
where $\pi^{i j} := \{x^i , x^j\}$ and
$$
  \omega = \mathbf{d}x^i \wedge \mathbf{d}\partial_i
  \,.
$$
The Hamiltonian cocycle from prop. \ref{CanonicalCocycleOfSymplecticLieAlgebroid}
is
$$
\pi= \frac{1}{2}\iota_v \iota_\epsilon  \omega=- \frac{1}{2} \pi^{ij} \partial_i \wedge \partial_j
$$
and the Chern-Simons element from prop. \ref{TheCSElement} is
$$
  \begin{aligned}
  \mathrm{cs}
  &= \iota_\epsilon \omega + \pi
  \\
  &= \partial_i \wedge \mathbf{d}x^i 
     - \frac{1}{2}\pi^{ij}\partial_i\wedge\partial_j
  \end{aligned}
  \,.
$$
In terms of $d_{\mathrm{W}}$ instead of $\mathbf{d}$ 
this is 
$$
  \begin{aligned}
    \mathrm{cs} & = \partial_i \wedge d_{\mathrm{W}}x^i -  \pi
    \\
    &= 
    \partial_i \wedge d_{\mathrm{W}}x^i + \frac{1}{2}\pi^{ij}\partial_i \partial_j\,.
      \end{aligned}
$$
So for $\Sigma$ a 2-manifold and
$$
  \Omega^\bullet(\Sigma) \leftarrow \mathrm{W}(\mathfrak{P}) : (X,\eta)
$$
a Poisson-Lie algebroid valued differential form on $\Sigma$
-- which in components is a function 
$X: \Sigma \to M$ and a 1-form 
$\eta \in \Omega^1(\Sigma, X^* T^* M)$ -- 
the corresponding AKSZ action is
$$
 \int_\Sigma \mathrm{cs}(X,\eta) 
   = \int_\Sigma
  \eta \wedge d_{\mathrm{dR}}X    
    + 
   \frac{1}{2}\pi^{ij}(X)\eta_i \wedge \eta_j
  \,.
$$
This is the Lagrangian of the Poisson $\sigma$-model 
\cite{Ikeda93,ScSt94,CattaneoFelder}.

\subsubsection{Courant $\sigma$-model}
\label{section.CourantSigmaModel}

A Courant algebroid is a symplectic Lie 2-algebroid. 
By the previous example this is a higher analog of a 
Poisson manifold. Expressed in components in the language of ordinary 
differential geometry, a Courant algebroid is a vector bundle $E$ 
over a manifold $M_{0}$, equipped with: a non-degenerate bilinear form
$\langle \cdot,\cdot \rangle$ on the fibers, a bilinear bracket
$[\cdot,\cdot]$ on sections $\Gamma(E)$, and a bundle map (called the
anchor) $\rho \colon E \to TM$, satisfying several compatibility
conditions. The bracket $[\cdot,\cdot]$ may be required to be
skew-symmetric (Def.\ 2.3.2 in \cite{RoytenbergCourant}), in which case it gives
rise to a Lie 2-algebra structure, or, alternatively, it may be required to
satisfy a Jacobi-like identity (Def.\ 2.6.1 in \cite{RoytenbergCourant}), in
which case it gives a Leibniz algebra structure.

It was shown in \cite{RoytenbergCourant} that Courant algebroids $E \to M_{0}$ 
in this component form
are in 1-1 correspondance with (non-negatively graded) grade 2
symplectic dg-manifolds $(M,v)$. Via this correspondance, $M$ is obtained
as a particular symplectic submanifold of $T^{\ast}[2]E[1]$ equipped
with its canonical symplectic structure.

Let $(M,v)$ be a Courant algebroid as above. In Darboux coordinates, the
symplectic structure is
\[
\omega = \mathbf{d}p_{i} \wedge \mathbf{d}q^{i} + \frac{1}{2}g_{ab} 
\mathbf{d}\xi^{a} \wedge  \mathbf{d}\xi^{b},
\]
with
\[
\deg{q^{i}}=0, ~ \deg{\xi^{a}}=1, ~ \deg{p_{i}}=2,
\]
and $g_{ab}$ are constants. The Chevalley-Eilenberg differential
corresponds to the vector field:
\[
v = P^{i}_{a} \xi^{a} \frac{\partial}{\partial q^{i}}
+ g^{a b} \bigl( P^{i}_{b}p_{i} - \frac{1}{2} T_{bcd}
\xi^{c} \xi^{d} \bigr) \frac{\partial}{\partial \xi^{a}}
+\left (-\frac{\partial P^{j}_{a}}{\partial q^{i}} \xi^{a}p_{j}
+ \frac{1}{6} \frac{\partial T_{abc}}{\partial q^{i}} \xi^{a} \xi^{b}
\xi^{c} \right)\frac{\partial}{\partial p_{i}}.
\]
Here $P^{i}_{a}=P^{i}_{a}(q)$ and $T_{abc}=T_{abc}(q)$ are particular degree zero
functions encoding the Courant algebroid structure.
Hence, the differential on the Weil algebra is:
\begin{align*}
d_{W} q^{i} &= P^{i}_{a} \xi^{a} + \mathbf{d} q^{i} \\
d_{W} \xi^{a} &= g^{a b} \bigl( P^{i}_{b}p_{i} - \frac{1}{2} T_{bcd}
\xi^{c} \xi^{d} \bigr) + \mathbf{d} \xi^{a} \\
d_{W} p_{i} &= -\frac{\partial P^{j}_{a}}{\partial q^{i}} \xi^{a}p_{j}
+ \frac{1}{6} \frac{\partial T_{abc}}{\partial q^{i}} \xi^{a} \xi^{b} \xi^{c}
+ \mathbf{d} p_{i}.
\end{align*}

Following remark. \ref{remark.local_hamiltonian}, we construct the
corresponding Hamiltonian cocycle from prop. 
\ref{CanonicalCocycleOfSymplecticLieAlgebroid}:
\begin{align*}
\pi &= \frac{1}{n+1}  \omega_{ab}\deg(x^a) x^a \wedge v^b\\
&= \frac{1}{3}\bigl( 2 p_{i} \wedge v(q^{i}) + g_{a b} \xi^{a}
\wedge v(\xi^{b}) \bigr)\\
&= \frac{1}{3}\bigl( 2 p_{i} P^{i}_{a} \xi^{a} + 
\xi^{a} P^{i}_{a}p_{i} - \frac{1}{2} T_{abc} \xi^{a}\xi^{b} \xi^{c} \bigr)\\
&= P^{i}_{a} \xi^{a} p_{i} - \frac{1}{6} T_{abc} \xi^{a}\xi^{b} \xi^{c}.
\end{align*}

The Chern-Simons element 
from prop. \ref{TheCSElement} is:
\begin{align*}
  \mathrm{cs} &= 
  \frac{1}{2}
  \left( 
    \sum_{a b} \deg(x^a) \,\omega_{a b} x^a \wedge d_{W}x^b  -  2 \pi
  \right)
   \\
   &= 
   p_{i} d_{W} q^{i} + \frac{1}{2}g_{ab} \xi^{a} d_{W} \xi^{b} -  \pi\\
   &=
   p_{i} d_{W} q^{i} + \frac{1}{2}g_{ab} \xi^{a} d_{W} \xi^{b}-P^{i}_{a} \xi^{a} p_{i} + \frac{1}{6} T_{abc} \xi^{a}\xi^{b} \xi^{c}.
\end{align*}
So for a Courant Lie 2-algebroid valued differential form datum 
\[
\Omega^\bullet(\Sigma) \leftarrow \mathrm{W}(\mathfrak{P}) : (X,A,F)
\]
on a closed 3-manifold $\Sigma$, we have
\[
  \int_\Sigma\mathrm{cs}(X,A,F)
  = 
 \int_\Sigma F_{i} \wedge d_{\mathrm{dR}} X^{i} +\frac{1}{2}g_{ab} A^{a}
  d_{\mathrm{dR}} A^{b}-P^{i}_{a} A^{a} F_{i}  
  + \frac{1}{6} T_{abc} A^{a} A^{b} A^{c}. 
\]
This is the AKSZ action for the Courant algebroid $\sigma$-model
from  \cite{Ikeda02, HoPa04, RoytenbergAKSZ}.

\subsubsection{Higher abelian Chern-Simons theory in $d = 4k+3$}
\label{HigherAbelianCSTheoryAsAKSZ}

For $k \in \mathbb{N}$, let $\mathfrak{a}$ be the 
delooping of the line Lie $2k$-algebra: $\mathfrak{a} = b^{2k+1}\mathbb{R}$.
By examples \ref{ExamplesOfInvariantPolynomials} there is, up to scale,
a unique binary invariant polynomial on $b^{2k+1}\mathbb{R}$,
and this is the wedge product of the unique generating 
unary invariant polynomial $\gamma$ in degree $2k+2$ with itself:
$$
  \omega
  := 
  \gamma \wedge \gamma
  \in 
  \mathrm{W}(b^{4k+4}\mathbb{R})
  \,.
$$
This invariant polynomial is clearly non-degenerate: 
for $c$ the canonical generator of $\mathrm{CE}(b^{2k+1}\mathbb{R})$
we have
$$
  \omega = \mathbf{d}c \wedge \mathbf{d}c
  \,.
$$ 
Therefore $(b^{2k+1}\mathbb{R}, \omega)$ induces an AKSZ $\sigma$-model
in dimension $n+1 = 4k+3$.
(On the other hand, on $b^{2k}\mathbb{R}$ there is only the 0 
binary invariant polynomial, so that no AKSZ-$\sigma$-models
are induced from $b^{2k}\mathbb{R}$.)

The Hamiltonian cocycle from 
proposition \ref{CanonicalCocycleOfSymplecticLieAlgebroid} vanishes
$$
  \pi = 0
$$ 
because the differential in the Chevalley-Eilenberg algebra $\mathrm{CE}(b^{2k+1}\mathbb{R})$ is trivial.
The Chern-Simons element from proposition \ref{TheCSElement} is
$$
  \mathrm{cs} = c \wedge \mathbf{d}c
  \,.
$$
A field configuration (definition \ref{definition.a-valued-differential-form})
$$
  \Omega^\bullet(\Sigma) \leftarrow
  \mathrm{W}(b^{2k+1})
  : C
$$ 
of this $\sigma$-model over a $(4k+3)$-dimensional manifold $\Sigma$ 
is simply a $(2k+1)$-form. The AKSZ action functional in this case
is
$$
  S_{AKSZ} : C \mapsto \int_\Sigma C \wedge d_{dR}C
  \,.
$$
The simplicity of this discussion is deceptive. In terms
of the outlook  in Section \ref{section.generalizations} below, it 
results from the fact that in AKSZ theory we are only looking at
$\infty$-Chern-Simons theory for universal Lie integrations and
for topologically trivial $\infty$-bundles. 
The $\infty$-Chern-Simons theory for $\mathfrak{a} = b^{2k+1}\mathbb{R}$
becomes nontrivial and rich when one drops these restrictions.
Then its configuration space is that of 
\emph{circle $(2k+1)$-bundles with connection} on $\Sigma$
(abelian $2k$-gerbes), classified
by ordinary differential cohomology in degree $2k+2$, and the 
action functional is given by the fiber integration along the projection to a point of the Beilinson-Deligne cup product in differential cohomology, which is locally
given by the above formula, but contains global twists. 
This is discussed in 
depth in \cite{HopkinsSinger}.

Once globalized this way, the above action functional is the 
action functional of higher $U(1)$-Chern-Simons theory
in dimention $4k+3$. In dimension 3 ($k = 0$) this is
discussed for instance in \cite{GuadagniniThuillier}
(notice that $U(1)$ is not simply connected, whence even in this
dimension there is a refinement of the standard story).
In dimension 7 ($k = 1$)  this higher Chern-Simons theory is the system whose
holographic boundary theory encodes the self-dual
2-form gauge theory on the fivebrane \cite{Witten}.
Generally, for every $k$ the $(4k+3)$-dimensional 
abelian Chern-Simons theory induces a self-dual higher
gauge theory holographically on its boundary, see \cite{BelovMoore}.

\section{Higher Chern-Simons field theory}
\label{section.CS_theory}

We indicate now the roots of the construction of 
the higher Chern-Simons action functionals discussed
above in a more encompassing general theory. 
We refer the reader to \cite{survey,FSSI} for details on this section.

The first step in this identification involves the
\emph{Lie integration} of an $L_\infty$-algebroid $\mathfrak{a}$ 
to a \emph{smooth $\infty$-groupoid} $\mathrm{exp}(\mathfrak{a})$
in analogy to how a Lie algebra integrates to a Lie group.
This in turn involves two aspects: the notion of a 
\emph{bare $\infty$-groupoid} on the one hand, and its 
\emph{smooth structure} 
on the other. 

Bare $\infty$-groupoids are presented by \emph{Kan complexes}: 
simplicial sets such that for all adjacent $k$-cells there
exists a composite $k$-cell, and such that every $k$-cell has
an inverse, up to $(k+1)$-cells, under this composition.
For instance for $G$ any ordinary groupoid there is such 
a Kan complex whose 0-cells are the objects of the groupoid,
and whose $k$-cells are the sequences of composable 
$k$-tuples of morphisms of the groupoid.

These bare $\infty$-groupoids are equipped with \emph{geometric structure}
by providing a rule for what the \emph{geometric families}
of $k$-cells in the $\infty$-groupoid are supposed to be. 
In this sense a smooth structure on an $\infty$-groupoid $A$ 
is given by declaring for each Cartesian space
$U = \mathbb{R}^n$ a set $A_k(U)$ of 
smooth $U$-parameterized families of $k$-morphisms in $A$, for
$k,n \in \mathbb{N}$. Collecting this data for all $k$
and all $U$ produces a functor
$$
  A : \mathrm{CartSp}^{\mathrm{op}} \to \mathrm{sSet}
$$
$$
  U \mapsto ([k] \mapsto A_k(U))
$$
from the opposite of the category of Cartesian spaces to 
the category of simplicial sets 
-- a \emph{simplicial presheaf} -- and this functor
encodes the structure of a smooth $\infty$-groupoid.

For instance if 
$A = (\xymatrix{A_1 \ar@<-2pt>[r]\ar@<+2pt>[r] & A_0})$ is an ordinary Lie groupoid, with a smooth manifold of objects $A_0$ and a smooth manifold of morphisms $A_1$, this assignment is given in the two lowest degrees  by sending $U$ to the set of smooth functions from $U$ to the
spaces of objects and morphisms: 
$A : U \mapsto C^\infty(U, A_k)$.
\begin{definition}
 \label{SmoothInfinityGroupoid}
A \emph{smooth $\infty$-groupoid} is a simplicial presheaf on the
category of Cartesian spaces and smooth functions between them,
\[
 A \colon \mathrm{CartSp}^{\mathrm{op}} \to \mathrm{sSet}
 \,,
\]
such that for each $U \in \mathrm{CartSp}$, the simplicial set
$A(U)$ is a Kan complex.

A morphism $f : A_1 \to A_2$ of smooth $\infty$-groupoids is
a morphism of the underlying simplicial presheaves 
(a natural transformation of functors).
A morphism is an \emph{equivalence of smooth $\infty$-groupoids} 
if it is stalkwise 
a 
weak homotopy equivalence of 
Kan complexes.
\end{definition}
\begin{remark}
  \label{CartSpAndGeneralizations}
Here the category of Cartesian spaces is just the
simplest of many possible choices. It can be varied at will, corresponding to which kind of geometric structure the 
$\infty$-groupoids are to be equipped with. For instance we 
can equivalently take instead the full category of 
smooth manifolds, without changing the notion of
smooth $\infty$-groupoid, up to equivalence. We could
also take richer categories, such as that of
smooth dg-manifolds. For \emph{non-positively}-graded
dg-manifolds we would speak of 
\emph{derived smooth $\infty$-groupoids} in this case.
These are necessary for discussion of the Lie integration of 
the full AKSZ BV-action, as opposed to just the grade-0 
functional that we concentrate on here.
\end{remark}
\begin{remark}
  It turns out that under the above notion of equivalence,
  \emph{every} simplicial presheaf is equivalent to 
  one that is objectwise a Kan complex. In a more
  abstract discussion than we want to get at here,
  we would more naturally say  that : 
  \emph{the 
  $\infty$-category of smooth $\infty$-groupoids is the simplicial localization
  $L_W [\mathrm{CartSp}^{\mathrm{op}}, \mathrm{sSet}]$
  at the stalkwise weak equivalences} (\cite{survey}).
\end{remark}
When regarded as simplicial presheaves on smooth test spaces, smooth $\infty$-groupoids have a canonical construction from 
$L_\infty$-algebroids by what is a parameterized 
version of the classical \emph{Sullivan construction}
in rational homotopy theory: the original construction 
\cite{Sullivan} sends 
-- in our $\infty$-Lie theoretic language -- an
$L_\infty$-algebroid $\mathfrak{a}$ 
to the simplicial set
\[
  \exp(\mathfrak{a})(*)
   :
[k] \mapsto
\mathrm{Hom}_{\mathrm{cdgAlg}_{\mathbb{R}}}(\mathrm{CE}(\mathfrak{a}),
\Omega^{\bullet}(\Delta^{k}))
  \,,
\]
whose $k$-cells are the flat $\mathfrak{a}$-valued differential
forms on the $k$-simplex (recall definition \ref{definition.a-valued-differential-form}).
It was noticed in \cite{Hinich, Getzler} 
(for the special case of $L_\infty$-algebras) that
this construction deserves to be understood as forming
the discrete $\infty$-groupoid that underlies the Lie integration of $\mathfrak{a}$.
In \cite{Henriques} the object $\exp(\mathfrak{a})$, 
still for the case that $\mathfrak{a} = b \mathfrak{g}$
comes from an $L_\infty$-algebra, is 
observed to be naturally equipped with a
Banach manifold structure. Moreover, a detailed discussion is
given showing that the \emph{truncations} $\tau_n \exp(\mathfrak{a})$
(the decategorification of the $\infty$-groupoid to an $n$-groupoid)
corresponds to the Lie integration to an $n$-group. For
instance for $\mathfrak{a} = b \mathfrak{g}$ coming from an
ordinary Lie algebra, $\tau_1 \exp(b \mathfrak{g})$ is
$\mathbf{B}G$: the one-object groupoid corresponding to the 
classical simply connected Lie group integrating $\mathfrak{g}$.
A detailed discussion of the smooth structure of
$\tau_1\exp(\mathfrak{a})$ for the case that $\mathfrak{a}$
is a Lie 1-algebroid was given in \cite{CrainicFernandes}.
There it is found that a certain cohomological obstruction
has to vanish in order that this is a genuine Lie groupoid
coming from a simplicial smooth manifold.
In \cite{TsengZhu} it was pointed out that however for
Lie 1-algebroids $\mathfrak{a}$ the 2-truncation
$\tau_2 \exp(\mathfrak{a})$ is always a simplicial manifold.

In \cite{FSSI} it was observed that 
without any assumption on $\mathfrak{a}$ and the 
truncation degree,
the construction always naturally -- and usefully -- extends
to smooth structure as encoded by presheaves on 
Cartesian test-spaces, simply by declaring 
the $U$-parameterized families
of $k$-cells in $\exp(\mathfrak{a})$ to be given by
$U$-parameterized families of flat $\mathfrak{a}$-valued 
connections:\begin{definition}
  \label{LieIntegration}
  For $\mathfrak{a}$ an $L_\infty$-algebroid, the functor
  $$
    \exp(\mathfrak{a}) 
      : 
    \mathrm{CartSp}^{\mathrm{op}}
      \to 
    \mathrm{sSet}
  $$
  to the category of simplicial sets is defined by setting, for
  $U \in \mathrm{CartSp}$ and $k \in \mathbb{N}$,
  $$
    \exp(\mathfrak{a})
    :
    (U, [k])
    \mapsto
    \left\{
      \xymatrix{
        \Omega^\bullet_{\mathrm{vert}, \mathrm{si}}(U \times \Delta^k)
        \ar@{<-}[r]^<<<<<<{A_{\mathrm{vert}}}
        &
        \mathrm{CE}(\mathfrak{a})
      }    
    \right\}
    \,,
  $$
  where $\Delta^k$ is the standard realization of the $k$-simplex 
  as a smooth manifold with boundary and corners, 
  and where
  $\Omega^\bullet_{\mathrm{vert}, \mathrm{si}}(U \times \Delta^k)$ is  
  the dg-algebra of \emph{vertical} differential forms on
  $U \times \Delta^k\to U$, that have \emph{sitting instants}
  towards the boundary faces of the simplex 
  (see \cite{FSSI} for details).
  \\
  We say that this \emph{simplicial presheaf}
  presents the \emph{universal Lie integration} of $\mathfrak{a}$.
 This can be understood as saying that the
Lie integration of $\mathfrak{a}$ always exists as a
\emph{diffeological} $\infty$-groupoid \cite{BaezHoffnung}.
\end{definition}

\begin{remark}
The simplicial presheaf $\exp(\mathfrak{a})$ can naturally be thought of as the presheaf of $U$-points of the simplicial set $[k] \mapsto
\mathrm{Hom}_{\mathrm{cdgAlg}_{\mathbb{R}}}(\mathrm{CE}(\mathfrak{a}),
\Omega^{\bullet}(\Delta^{k}))$ described above. Indeed, the 
dg-algebra of vertical differential forms on $U \times \Delta^k$ is naturally isomorphic to $\mathrm{CE}(U\times \mathfrak{T}\Delta^k)$. Also note that this is in turn isomorphic to the (completed) tensor product $\mathrm{CE}(U)\hat{\otimes}\mathrm{CE}(\mathfrak{T}\Delta^k)=C^\infty(U)\hat{\otimes}
\Omega^\bullet(\Delta^k)$.
\end{remark}
Indeed, the simplicial presheaf given in 
definition \ref{LieIntegration}
is a Lie $\infty$-groupoid in the sense of 
definition \ref{SmoothInfinityGroupoid}.
\begin{proposition}
For $\mathfrak{a}$ an $L_\infty$-algebroid, the simplicial
presheaf $\exp(\mathfrak{a})$ is a Lie $\infty$-groupoid
(is objectwise a Kan complex).
\end{proposition}
\begin{proof}
 Since the differential forms in the above
definition are required to have sitting instants,  
they can be smoothly pulled back along the
standard continuous retract
projections $\Delta^n \to \Lambda^n_i$
of the $n$-horns, because these are smooth
away from the boundary. This provides horn fillers in the
standard way.
See also the proof of Proposition  4.2.10 in \cite{FSSI}.
\end{proof}
\begin{remark}
While it can be useful in specific computations to know
that $\exp(\mathfrak{a})$ is degreewise a smooth manifold, 
if indeed it is, no
general concept in smooth higher geometry requires this
assumption. On the other hand, one can show 
\cite{survey}, that every smooth $\infty$-groupoid
is equivalent to a simplicial presheaf that is degreewise
a disjoint union of smooth manifolds, even to one
that is degreewise a disjoint union of Cartesian spaces. 
\end{remark}
\begin{remark}
  A category with weak equivalences, such as
  that 
  of smooth $\infty$-groupoids, 
  is canonically equipped with a \emph{derived hom-functor},
  which to smooth $\infty$-groupoids $X$ and $\exp(\mathfrak{a})$
  assigns an $\infty$-groupoid
  $\mathbf{RHom}(X, \exp(\mathfrak{a}))$.
  One finds that the objects of this $\infty$-groupoid are
  {\v C}ech cocycles for
  \emph{principal $\infty$-bundles} $P \to X$ that are modeled on 
  $\mathfrak{a}$ in higher analogy of how an ordinary 
  smooth principal bundle is ``modeled on'' the Lie algebra
  of its structure group. The 1-morphisms in $\mathbf{RHom}(X,    
\exp(\mathfrak{a}))$ are the gauge transformations of these principal $\infty$-bundles, and so on 
  \cite{NSS, survey}.
\end{remark}
Note that in the definition of $\exp(\mathfrak{a})$ only the Chevalley-Eilenberg algebra of $\mathfrak{a}$ is relevant. The Weil algebra $\mathrm{W}(\mathfrak{a})$ is then introduced in order to describe a differential refinement of $ \exp(\mathfrak{a})$.
\begin{definition}
  For $\mathfrak{a}$ an $L_\infty$-algebroid 
  write $\exp(\mathfrak{a})_{\mathrm{diff}}$ for the 
  simplicial presheaf given by
  $$
    \exp(\mathfrak{a})_{\mathrm{diff}}
    : 
    (U, [k]) \mapsto
    \left\{
       \raisebox{20pt}{
      \xymatrix{
         \Omega^\bullet_{\mathrm{vert,si}}(U \times \Delta^k)_{\mathrm{vert}}
         \ar@{<-}[r]^<<<<{A_\mathrm{vert}} & 
         \mathrm{CE}(\mathfrak{a})
         \\
         \Omega^\bullet(U \times \Delta^k)\ar[u]
         \ar@{<-}[r]^<<<<<<{(A,F_A)} & 
         \mathrm{W}(\mathfrak{a})\ar[u]         
      }
      }
    \right\} 
    \,,
  $$
  where on the right we have the set of horizontal
  dg-algebra homomorphims that make the square commute, 
  as indicated.
\end{definition}
Notice that by definition \ref{definition.a-valued-differential-form} 
the bottom horizontal morphisms on the right are
$\mathfrak{a}$-valued differential forms on 
$U \times \Delta^k$.
\begin{proposition}
  The canonical projection morphism 
  $$
    \exp(\mathfrak{a})_{\mathrm{diff}} \to \exp(\mathfrak{a})
  $$
  to the Lie integration from definition \ref{LieIntegration}
  is an equivalence of smooth $\infty$-groupoids.
\end{proposition}
\begin{remark}
It is via this property that the Weil algebra serves in 
$\infty$-Chern-Weil theory as part of a 
\emph{resolution} of $\exp(\mathfrak{a})$ from which
curvature characteristics are built. 
\end{remark}
Let now $\langle-\rangle$ be an invariant polynomial on $\mathfrak{a}$ (
definition \ref{InvariantPolynomial}). 
Evaluating $\langle-\rangle$ on the curvature 
$F_A$ of an $\mathfrak{a}$-connection $A$ gives a closed differential form $\langle F_A\rangle$ on $U\times \Delta^k$,
according to definition \ref{definition.chern-simons-form}. 
This differential form, however, will in general not descend to the base space $U$. This naturally leads to considering the following definition, which picks the universal subobject of 
$\exp(\mathfrak{a})_{\mathrm{diff}}$ that makes
all \emph{curvature characteristic forms} 
$\langle F_A\rangle$  descend to base space.
\begin{definition}
  Define the simplicial presheaf
  $$
    \exp(\mathfrak{a})_{\mathrm{conn}}
    :
    (U,[k])
    \mapsto
    \left\{
      \raisebox{37pt}{
      \xymatrix{
        \Omega^\bullet_{\mathrm{vert,si}}(U \times \Delta^k)
        \ar@{<-}[r]^<<<<<{A_{\mathrm{vert}}}
        &
        \mathrm{CE}(\mathfrak{a})
        \\
        \Omega^\bullet(U \times \Delta^k)
        \ar[u]
        \ar@{<-}[r]^<<<<<<{A}
        &
        \ar[u]
        \mathrm{W}(\mathfrak{a})
        \\
        \Omega^\bullet(U)_{\mathrm{closed}}
        \ar[u]
        \ar@{<-}[r]^<<<<<<{\langle F_A \rangle}
        &
        \mathrm{inv}(\mathfrak{a})
        \ar[u]
      }
      }
    \right\}
    \,,
  $$
  where on the right we have the set of $\mathfrak{a}$-valued forms
  $A$ on $U \times \Delta^k$ that make the diagram commute
  as indicated.
\end{definition}
This has the following interpretation
(first considered in \cite{SSSI}):
\begin{enumerate}
  \item 
    The commutativity of the top diagram say that 
    the $\mathfrak{a}$-valued differential form 
    $A$ on $U \times \Delta^k$ 
    is \emph{vertically flat} with respect to the
    trivial simplex bundle $U \times \Delta \to U$.
    This is an analogue of the verticality condition
    for an ordinary \emph{Ehresmann connection}.
  \item
   The commutativity of the lower diagram says that 
   all curvature forms $F_A$ \emph{transform covariantly} along the
   simplices in such a way as to make all the curvature characteristic
   form $\langle F_A\rangle$ descent to base space. This 
   is the analogue of the horizonatlity condition on an 
   ordinary Ehresmann connection.
\end{enumerate}
One finds therefore that 
for $X$ a smooth manifold, an element in $\mathbf{RHom}(X, \exp(\mathfrak{a})_{\mathrm{conn}})$
is
\begin{enumerate}
  \item a choice of good open cover 
  $\{U_i \to X\}$;
  \item
    on each patch $U_i$ differential form 
    data $A_i$ with 
    values in $\mathfrak{a}$;
  \item
    on each double intersection a choice of 1-parameter
    gauge transformation between the corresponding differential form data;
  \item
    on each triple intersection a choice of 2-parameter gauge-of-gauge
    transformation;
   \item and so on.
\end{enumerate}
Such a \emph{differential {\v C}ech cocycle} is essentially what 
defines an $\infty$-connection on a principal $\infty$-bundle. 
This is discussed
in detail in \cite{FSSI, survey}. 
\begin{remark}
Since for the discussion of the simple case of 
AKSZ $\sigma$-models we can 
assume that the underlying $\infty$-bundle is 
\emph{trivial},  only a single 0-simplex
$$
 \raisebox{37pt}{
      \xymatrix{
        C^\infty(\Sigma)
        \ar@{<-}[r]^{A_{\mathrm{vert}}}
        &
        \mathrm{CE}(\mathfrak{a})
        \\
        \Omega^\bullet(\Sigma)
        \ar[u]
        \ar@{<-}[r]^{A}
        &
        \ar[u]
        \mathrm{W}(\mathfrak{a})
        \\
        \Omega^\bullet(\Sigma)_{\mathrm{closed}}
        \ar[u]
        \ar@{<-}[r]^{\langle F_A \rangle}
        &
        \mathrm{inv}(\mathfrak{a})
        \ar[u]
      }
      }
$$
is involved in the description of the AKSZ $\sigma$-model.
\end{remark}
With these concepts in hand, we can 
now explain how the datum  of a triple 
$(\mu,\mathrm{cs},\langle-\rangle)$ consisting 
of a Chern-Simons element witnessing the transgression between an invariant polynomial and a cocyle 
(definition \ref{TransgressionAndCSElements})
serves to present a 
\emph{differential characteristic class}
in terms of a morphism of smooth $\infty$-groupoids. 
To see this, recall that the line delooping 
$L_\infty$-algebroid $b^{n+1}\mathbb{R}$ 
of the Lie line $(n+1)$-algebra is defined by the fact that 
   $\mathrm{CE}(b^{n+1}\mathbb{R})$ is generated over 
   $\mathbb{R}$ from a single generator in degree $n+1$
   with vanishing differential. As an immediate consequence, an $(n+1)$-cocylce $\mu$ on an $L_\infty$-algebroid $\mathfrak{a}$ is the same thing as a dg-algebra morphism
  $$
   \mu : \mathrm{CE}(b^{n+1}\mathbb{R}) \to \mathrm{CE}(\mathfrak{a}).
  $$   
 Similarly, a triple  $(\mu,\mathrm{cs},\langle-\rangle)$ is naturally identified with a commutative diagram of dg-algebras:
   $$
    \xymatrix{
      \mathrm{CE}(\mathfrak{a})
      \ar@{<-}[r]^{\mu}
      &
      \mathrm{CE}(b^{n+1} \mathbb{R})
      \\
      \mathrm{W}(\mathfrak{a})
      \ar@{<-}[r]^{\mathrm{cs}}
      \ar[u]
      &
      \mathrm{W}(b^{n+1} \mathbb{R})
      \ar[u]
      \\
      \mathrm{inv}(\mathfrak{a})
      \ar@{<-}[r]^{\langle - \rangle}
      \ar[u]
      &
      \mathrm{inv}(b^{n+1} \mathbb{R})
      \ar[u]
    }
    \,.
  $$
 Pasting this diagram to the
  one above defining $\exp(\mathfrak{a})_{\mathrm{conn}}$ 
  leads to the following observation, discussed in \cite{FSSI}.
\begin{proposition}
  \label{CWPresentation}
  Every triple $(\mu, \mathrm{cs}, \langle -\rangle)$ 
  induces a morphism
  $$
    \exp(\mathfrak{a})_{\mathrm{conn}}
    \to
    \exp(b^{n+1}\mathbb{R})_{\mathrm{conn}}\,.
  $$
  This morphism is in fact the presentation of the 
$\infty$-Chern-Weil homomorphism induced by the 
invariant polynomial $\langle -\rangle$.
\end{proposition} 
This means that for
$$
  (\nabla : X \to \exp(\mathfrak{a})_{\mathrm{conn}})
  \in 
  \mathbf{RHom}(X, \exp(\mathfrak{a})_{\mathrm{conn}})
$$
an $\mathfrak{a}$-valued $\infty$-connection, the composite
$$
  X \xrightarrow{\nabla}
  \exp(\mathfrak{a})_{\mathrm{conn}}
    \xrightarrow{\exp(\mathrm{cs})}
    \exp(b^{n+1}\mathbb{R})_{\mathrm{conn}}
$$
is a representative of the curvature $(n+1)$-form on $X$
that the $\infty$-Chern-Weil homomorphism induced by
$\langle-\rangle$ assigns to $\nabla$.

We can now formalize the observation mentioned in the
introduction, that the Chern-Weil homomorphism
plays the role of action functional for $\sigma$-model
quantum field theories. Indeed, in view of the above constructions, the AKSZ $\sigma$-model
Lagrangian corresponds to forming the following
commutative diagram:
\[
\xymatrix{
      C^\infty(\Sigma)
            & 
      \mathrm{CE}(\mathfrak{P})
        \ar[l]_{\phi_{\mathrm{vert}}} 
        & 
        \mathrm{CE}(b^{n+1}\mathbb{R})\ar[l]_{\frac{1}{n}\pi}
        & 
        : 
       \frac{1}{n}\pi(\phi_{\mathrm{vert}})
       \\
       \Omega^\bullet(\Sigma)
        \ar[u]
        &
        \mathrm{W}(\mathfrak{P})
        \ar[u]\ar[l]_{\phi} 
        & 
        \mathrm{W}(b^{n+1}\mathbb{R})
        \ar[l]_{\mathrm{cs}}\ar[u]
        &
        : \mathrm{cs}(\phi)
       \\
      \Omega^\bullet(\Sigma)_{\mathrm{closed}}
       \ar[u]
       &
       \mathrm{inv}(\mathfrak{P})
        \ar[l]_{F_\phi}
        \ar[u] 
        & 
       \mathrm{inv}(b^{n+1}\mathbb{R})
        \ar[u]\ar[l]_{\omega}
       &
       : \omega(F_\phi)
}
\,.
\]
In other words, under the identification of the AKSZ action functional
as an instance of the $\infty$-Chern-Weil homomorphism
we indeed translate concepts as shown in the table in the introduction:
the symplectic form is the invariant polynomial
that induces the Chern-Weil homomorphism, the Hamiltonian
is the cocycle that it transgresses to, and the Chern-Simons
element that witnesses the transgression is the Lagrangian.
\par
This suggests the following general definition of a higher Chern-Simons field theory.
\begin{definition}
  \label{LagrangianFromCW}
  Let $\Sigma$ be an $(n+1)$-dimensional compact smooth manifold
  and
  $\mathfrak{a}$ an $L_\infty$-algebroid equipped with an 
  invariant polynomial $\langle - \rangle$.
  Let $\mathrm{cs}$ be a Chern-Simons element witnessing 
  its transgression to a cocycle $\mu$. Then way may say
  \begin{itemize}
    \item
      A morphism $\phi : \Sigma \to  \exp(\mathfrak{a})_{\mathrm{conn}}$
      is a \emph{field configuration} on $\Sigma$ with values in 
      $\mathfrak{a}$.
    \item
      The assignment
      $$
        \phi \mapsto \mathrm{cs}(\phi) \in \Omega^{n+1}(\Sigma)
      $$
      is the \emph{Lagrangian} defined by $\mathrm{cs}$;
    \item
      The assignment
      $$
        \phi \mapsto \int_\Sigma\mathrm{cs}(\phi) \in \mathbb{R}      
      $$
      is the \emph{action functional} defined by $\mathrm{cs}$.
  \end{itemize}
  The collection of these notions we call the
  \emph{higher Chern-Simons field theory} defined by 
  $\mathrm{cs}$.
\end{definition}

\section{Generalizations}
\label{section.generalizations}

The identification of the AKSZ action functionals
as a special case of the general abstract 
Chern-Weil homomorphism allows to tranfer various
insights about the general theory and about its 
other special cases to AKSZ theory. 
We close this article by briefly indicating a few.
More detailed discussion shall be given elsewhere.

\subsection{Symplectic $n$-groupoids and nontrivial topology}
   
   The smooth $\infty$-groupoid $\exp(\mathfrak{P})$
   that integrates a symplectic Lie $n$-algebroid 
   (as discussed in \ref{section.CS_theory})
   is the ``higher universal'' Lie integration of $\mathfrak{P}$.
   One finds (see \cite{survey} for the discussion
   in the case of smooth $\infty$-groupoids, following
   the discussion of Banach-$\infty$-groupoids in 
   \cite{Henriques}) that its geometric realization
   in topological spaces is \emph{$\infty$-connected}
   (hence: contractible) in analogy to how the classical 
   universal Lie integration of a Lie algebra is 
   1-connected (hence: simply connected).

   As we have shown here, this is sufficient for the 
   traditional description of AKSZ $\sigma$-models. 
   But more generally one will be interested in the 
   universal integration to the smooth $n$-groupoid
   $P := \tau_n \exp(\mathfrak{P})$ obtained as the
   \emph{$n$-truncation} (where one retains only
   equivalence classes of $n$-cells in $\exp(\mathfrak{P})$).

   For instance for $\mathfrak{P} = b\mathfrak{g}$ the delooping
   of a semisimple Lie algebra $\mathfrak{g}$
   (the case of the Courant Lie 2-algebroid over the point)
   we have that $\tau_1 \exp(b \mathfrak{g}) \simeq \mathbf{B}G$
   is the one-object Lie groupoid obtained from the simply-connected
   Lie group that integrates $\mathfrak{g}$, while the untruncated
   $\exp(b\mathfrak{g})$ is some higher extension of $\mathbf{B}G$
   by higher abelian $\infty$-groups.
   
   This truncation, however, also affects the coefficient object
   of the $\infty$-Chern-Weil homomorphism (discussed in detail in
   \cite{FSSI}): notably the untruncated AKSZ action functional
   $$
     \exp(\mathrm{cs}_\omega) 
      : 
     \exp(\mathfrak{P})_{\mathrm{conn}}
     \to 
     \exp(b^{n+1}\mathbb{R})_{\mathrm{conn}}
     \simeq
     \mathbf{B}^{n+1}\mathbb{R}_{\mathrm{conn}}
   $$
   descends to the truncation only up to a quotient by the
   group $K \subset \mathbb{R}$ of \emph{periods} of 
   the hamiltonian cocycle $\pi$:
   $$
     \exp(\mathrm{cs}_\omega) 
      : 
     \tau_n\exp(\mathfrak{P})_{\mathrm{conn}}
     \to 
     \mathbf{B}^{n+1}\mathbb{R}/K_{\mathrm{conn}}
     \,.
   $$
   Typically we have $K \simeq \mathbb{Z}$ and hence
   $\mathbb{R}/K \simeq U(1)$. This way
   the properly truncated AKSZ action functional 
   indeed takes values in circle $n+1$-bundles. This 
   becomes 
   a further \emph{quantization condition} for the 
   field configurations, discussed in the next item.

\subsection{$\infty$-Connections on nontrivial $\mathfrak{a}$-principal
    $\infty$-bundles: AKSZ instantons}
    
   As we have shown in this article, the fields of AKSZ $\sigma$-model
   theories may be understood as 
   $\infty$-connections on \emph{trivial}
   $\exp(\mathfrak{P})$-principal $\infty$-bundles
   or, by the previous paragraph, on trivial 
   $\tau_n \exp(\mathfrak{P})$-principal $\infty$-bundles.

   The general theory of \cite{FSSI, survey}
   provides also a description of $\infty$-connections
   on non-trivial such $\infty$-bundles and there is
   no reason to restrict attention to the
   trivial ones. Alternatively, as propagated by Kotov and Strobl,
one can use the notion of (possibly non-trivial) Q-bundles, in which connections turn out to be
sections in the category of graded manifolds \cite{KoSt07}. The relation of our $\infty$-bundles to their Q-bundles
generalizes the one of an ordinary principal bundle to its associated Atiyah algebroid.
   Such fields with non-trivial underlying
   principal $\infty$-bundles
   correspond to what in the analog situation of 
   Yang-Mills theory are called
   \emph{instanton} field configurations. 
   These are of importance in a comprehensive discussion of the
   quantum theory.
   
   This issue plays only a minor role in low
   dimensions. For instance the reason that the fields
   of Chern-Simons theory are and can be taken to be
   connections on trivial $G$-principal bundles 
   on $\Sigma$ is that
   for simply connected Lie groups $G$ the
   classifying space $B G$ has its first non-trivial
   homotopy group in degree 4, so that all $G$-principal
   bundles on a 3-dimensional $\Sigma$ are necessarily
   trivializable. 

   But by the same argument there are inevitably 
   AKSZ instanton contributions
   from fields that are connections on non-trivial $\infty$-bundles
   as soon as we pass to 4-dimensional AKSZ models and beyond.   

\subsection{Twisted AKSZ-structures and higher extensions of 
     symplectic $L_\infty$-algebroids}
     
   For any differential characteristic class
   $$
     \hat {\mathbf{c}} : A_{\mathrm{conn}} \to 
      \mathbf{B}^{n+1} \mathbb{R}/K_{\mathrm{conn}}
   $$
   such as obtained from Lie integration of 
   a Chern-Simons element:
   $$
     \exp(\mathrm{cs}) : \tau_n \exp(\mathfrak{a})_{\mathrm{conn}}
       \to 
      \mathbf{B}^{n+1}\mathbb{R}/K
   $$
   it is of interst to study the \emph{homotopy fibers}
   that this induces on cocycle $\infty$-groupoids over
   a given base space $X$. In \cite{SSSIII, FSSI} the $\infty$-groupoid
   $\hat {\mathbf{c}}\mathrm{Struc}(X)$ of
   \emph{twisted $\hat {\mathbf{c}}$}-structures is
   introduced as the homotopy pullback
   $$
     \xymatrix{
       \hat {\mathbf{c}}\mathrm{Struc}_{\mathrm{tw}}(X)
       \ar[r]^{\mathrm{tw}}
       \ar[d]
       &
       H^n_{\mathrm{diff}}(X,K)
       \ar[d]
       \\
       \mathbf{RHom}(X,A_{\mathrm{conn}})
       \ar[r]^<<<<<{\hat {\mathbf{c}}}
       &
       \mathbf{RHom}(X,\mathbf{B}^{n+1} K_{\mathrm{conn}})
     }
     \,,
   $$
   where the right vertical morphisms -- unique up to equivalence --
   picks one cocycle representative in each cohomology class.
   The morphism $\mathrm{tw}$ sends a given twisted differential
   cocycle to its \emph{twist}. The fibers over the trivial
   twist are precisely the $\widehat A$-principal $\infty$-bundles
   with connection, where $\widehat A$ is the extension of $A$ 
   classified by $\mathbf{c}$, 
   wich is characterized by the fact that
   it sits in a fiber sequence
   $$
     \mathbf{B}^{n} K \to \widehat A \to A 
     \,.
   $$
   
   In \cite{FSSI} this is discussed in detail for the case that 
   $\mathbf{c} = \frac{1}{2}\mathbf{p}_1$ is a smooth 
   refinement of the first fractional Pontryagin class 
   and for the case $\mathbf{c} = \frac{1}{6}\mathbf{p}_2$
   of the fractional second Pontryagin class. 
   In these cases the extension $\hat A$ is the delooping,
   respectively, of the smooth \emph{string 2-group}
   and of the smooth \emph{fivebrane 6-group}.
   The corresponding \emph{twisted differential string-structures}
   and \emph{twisted differential fivebrane structures}
   are shown there (following \cite{SSSIII}) 
   to encode the Green-Schwarz mechanism in 
   heterotic string theory and dual heterotic string theory,
   respectively.
   
   By our discussion here, all these constructions have their
   natural analogs for AKSZ $\sigma$-models, too. In particular,
   every symplectic Lie $n$-algebroid $(\mathfrak{P}, \omega)$
   with Hamiltonian $\pi \in \mathrm{CE}(\mathfrak{P})$
   has a canonical (``string-like'') extension
   $$
     b^{n} \mathbb{R} \to \widehat {\mathfrak{P}} \to \mathfrak{P}
   $$
   classified by $\pi$ regarded  
   as an $L_\infty$-cocycle $\pi : \mathfrak{P} \to b^{n+1}\mathbb{R}$.
   
   This extension is easy to describe: the Chevalley-Eilenberg algebra
   $\mathrm{CE}(\widehat {\mathfrak{P}})$ is that of $\mathfrak{P}$
   with a single generator $b$ in degree $n$ adjoined, and the
   differential extended to this generator by the formula
   $$
     d_{\mathrm{CE}(\widehat {\mathfrak{P}})} : b \mapsto
     \pi 
     \,.
   $$
   
   A \emph{twisted differential $\exp(\mathfrak{P})$-structure}
   is accordingly an $\exp(\mathfrak{P})$-$\infty$-connection $\phi$
   (an AKSZ $\sigma$-model field) equipped with an
   equivalence of its curvature characteristic $\omega(\phi)$
   to a presribed ``twisting class''. When the twisting class is
   trivial, then these are equivalently 
   $\exp(\widehat {\mathfrak{P}})$-principal $\infty$-connections.
   
   Notice that these considerations are relevant only over
   a base space of dimension at least $(n+2)$. Compare this
   again to the familiar case of Chern-Simons theory: 
   in Chern-Simons theory itself the $G$-principal bundle
   may always be taken to be trivial, since the base space
   $\Sigma$ is taken to be 3-dimensional. But all the 
   constructions of Chern-Simons theory make sense also
   over arbitrary $X$. Generally, the Chern-Weil homomorphism
   assigns a  \emph{Chern-Simons circle 3-bundle} to every
   suitabe $G$-principal bundle on $X$, and its volume holonomy
   is the corresponding Chern-Simons functional for this situation.
   Analogously one can consider AKSZ theory over higher dimensional
   base spaces.

\subsection{Relation to higher dimensional supergravity}

    AKSZ theory is not the only class of field theories where
    it was noticed that field configurations 
    have an 
    interpretation in terms of morphisms of dg-algebras. 
    Almost two decades earlier 
    originates the observation that 
    (higher dimensional) \emph{supergravity} 
    (see for instance \cite{DeligneMorgan} for standard itroductions)
    has a rather beautiful description in such terms.
    A detailed exposition of this dg-algebraic 
    approach to supergravity is in 
    the textbook \cite{CDF}.

    The authors there speak of 
    ``free differential algebras''  (``FDAs'') where they would mean 
    what in the mathematical literature are called
    ``quasi-free dg-algebras'' or ``semi-free dg-algebras'' 
    -- those whose underlying graded
    algebra is free, as in our definition \ref{LInfinityAlgebroids} 
    of Chevalley-Eilenberg algebras. 
    Moreover, what we here observe are morphisms
    out of the Weil algebra 
    (definition \ref{definition.a-valued-differential-form}) 
    they call ``soft group manifolds''.
    But apart from these purely terminological differences
    one finds that
    the observations that drive the developments there are
    precisely the following, here reformulated in our 
    $\infty$-Lie theoretic language (see section 4 of \cite{survey}):

    First of all it is a standard fact that in ``first order formulation''
    the field configurations of gravity in $d+1$ dimensions 
    are naturally presented by
    $\mathfrak{iso}(d,1)$-vaued connection forms, where
    $\mathfrak{iso}(d,1) = \mathbb{R}^{d+1} \ltimes \mathfrak{so}(d,1)$
    is the Poincar{\'e}-Lie algebra.
    This perspective is inevitable in the context of 
    supergravity, where the first-order formulation is 
    required by the coupling to fermions.
    There is an evident super-algebra generalization
    of the Poincar{\'e} Lie algebra to the 
    \emph{super-Poincar{\'e}}-Lie algebra
    $\mathfrak{siso}(d,1)$ and a field configuration
    of supergravity is an $\mathfrak{siso}(d,1)$-valued connection.
    Or rather, such a connection encodes the \emph{graviton} field and 
    its superpartner field, the \emph{gravitino} field,
    but not yet the higher bosonic form fields generically
    present in higher supergravity theories. 
    The first central observation of \cite{DAuriaFre} is
    (in our words) that these
    naturally appear after  passage to 
    \emph{$L_\infty$-extensions}
    of $\mathfrak{siso}(d,1)$.
    
    To put this statement into our context,
    notice that there is a fairly straightforward
    super-geometric extension of general abstract 
    higher Chern-Weil theory 
    in which 
    $L_\infty$-algebroids are generalized to \emph{super}
    $L_\infty$-algebroids as Lie algebras are generalized to 
    super Lie algebras (see section 3.5 of \cite{survey}). 
    
    It turns out that super-Lie algebra
    cohomology of $\mathfrak{siso}(d,1)$  
    contains a certain finite number of \emph{exceptional}
    cocycles $\mu : \mathfrak{siso}(d,1) \to b^{n+1}\mathbb{R}$
    for certain values of $d$, whose existence is naturally
    understood from the existence of the four normed division
    algebras (\cite{Huerta}). 
    In particular, for $d = 10$
    there is a 4-cocycle $\mu_4 : \mathfrak{siso}(10,1) \to 
    b^4 \mathbb{R}$. The Lie 3-algebra extension that it classifies
    $$
      b^3 \mathbb{R} \to \mathfrak{sugra}_{11} \to 
      \mathfrak{siso}(10,1)
    $$
    has been called the \emph{supergravity Lie 3-algebra}
    in \cite{SSSI}. It turns out (\cite{DAuriaFre}) 
    that this carries, in turn,
    a 7-cocycle 
    $\mu_7 : \mathfrak{sugra}_{11} \to b^7 \mathbb{R}$.
    The original observation of \cite{DAuriaFre}
    was (not in these words, though, but easily translated
    into $\infty$-Lie theoretic terms
    using  our discussion here) that 
    11-dimensional supergravity, including its higher form
    field degrees of freedom, is naturally understood
    as a theory of $\infty$-connections with values in the
    corresponding \emph{supergravity Lie 6-algebra}
    $$
      b^7 \mathbb{R} \to \widehat{\mathfrak{sugra}}_{11}
        \to \mathfrak{sugra}_{11}
    $$
    and that the construction of its action functional
    is governed by the higher Lie theory of this object.
    
    While 11-dimensional supergravity is not 
    entirely a higher Chern-Simons-theory, it 
    crucially does involve Chern-Simons terms in its
    action functionals. Indeed, one can see that 
    one of the characterizing conditons on 
    a supergravity action functional -- the 
    one called the \emph{cosmo-cocycle condition}
    in \cite{CDF} -- is the defining condition
    on a Chern-Simons element in our def. 
   \ref{TransgressionAndCSElements}, but solved only up
   to first order in the curvature terms.
   It may be noteworthy in this context that there are
   various speculations 
   (see \cite{CSSupergravity} for discussion and review)
   that higher dimensional supergravity should be thought
   of as a limiting theory of a genuine higher Chern-Simons theory.
   
   This shows that there is a close conceptual relation between
   AKSZ $\sigma$-models, higher Chern-Simons theories
   and higher dimensional supergravity,
   mediated by abstract higher Chern-Weil theory.
   The various implications of this observation shall
   be explored elsewhere.
  
\addcontentsline{toc}{section}{References}

\end{document}